\documentclass[11pt]{scrartcl}

\usepackage{float}
\usepackage[singlespacing]{setspace} 
\usepackage{sectsty} 
\usepackage{listings}
\usepackage{tabularx} 
\usepackage{csquotes} 
\usepackage{cite} 
\usepackage{tocloft} 
\usepackage{bibunits} 
\usepackage{algorithm,algpseudocode} 
\usepackage{afterpage}

\usepackage{color,xcolor}
\definecolor{color1}{HTML}{3c1b64}
\definecolor{color1L}{HTML}{f8f5fc}
\definecolor{color2}{HTML}{FF6C00}
\definecolor{color2L}{HTML}{fff7f2}
\definecolor{color3}{HTML}{A0204C}
\definecolor{color3L}{HTML}{FDF1F5}
\definecolor{color0}{HTML}{949494}
\definecolor{color0M}{HTML}{DCDCDC}
\definecolor{color0L}{HTML}{F8F8F8}
\subsectionfont{\normalsize}  
\usepackage{enumitem}
\usepackage[color=color2L]{todonotes}
\usepackage[]{hyperref} 
\hypersetup{colorlinks=true,linkcolor=color2, citecolor=color3,urlcolor=color1}
\usepackage{titlesec, blindtext}

\usepackage{scalerel,stackengine,amsmath} 
\usepackage{mathtools}
\usepackage{amssymb} 
\usepackage{braket  }
\usepackage{amsthm} 
\usepackage{bbm} 
\usepackage{bbold} 
\DeclareFixedFont{\ttb}{T1}{txtt}{bx}{n}{12} 
\DeclareFixedFont{\ttm}{T1}{txtt}{m}{n}{12}  
\DeclareMathOperator{\tr}{tr}

\DeclareMathOperator{\can}{\mathrm{CAN}}

\allowdisplaybreaks
\newcommand{\mathcolorbox}[1]{\fcolorbox{color2}{color2L}{$\displaystyle #1$}}
\usepackage{nameref,cleveref}
\makeatletter
\newtheorem*{rep@theorem}{\rep@title}
\newcommand{\newreptheorem}[2]{%
\newenvironment{rep#1}[1]{%
\def\rep@title{#2 \ref{##1}}%
\begin{rep@theorem}}%
{\end{rep@theorem}}}
\makeatother
\newtheorem{prop}{Proposition}
\newreptheorem{prop}{Proposition}
\crefname{section}{Section}{Sections}
\Crefname{section}{Section}{Sections}
\crefname{algorithm}{Algorithm}{Algorithm}
\Crefname{algorithm}{Algorithm}{Algorithm}
\crefname{equation}{Eq.}{Eq.}
\Crefname{equation}{Eq.}{Eq.}
\crefname{figure}{Fig.}{Fig.}
\Crefname{figure}{Fig.}{Fig.}
\crefname{appendix}{Appendix}{Appendices}
\Crefname{appendix}{Appendix}{Appendices}
\crefname{prop}{Proposition}{Propositions}
\Crefname{prop}{Proposition}{Propositions}

\usepackage{graphicx}
\usepackage{tikz}
\usepackage{pgfplots}
\usepackage{subcaption}
\usepackage{adjustbox} 		
\usepackage[section]{placeins} 
\usetikzlibrary{quantikz,matrix,fit,calc,shapes.geometric,backgrounds,positioning,decorations.markings,decorations.pathreplacing,arrows,knots,hobby,angles,quotes,shapes,snakes,automata,petri}
\usepackage{wrapfig} 
\usepackage{hhline,colortbl} 
\tikzset{
perceptron0/.style = {circle,draw=color0,line width=1pt,fill=color0L,minimum width=0.7cm},
perceptronS0/.style = {circle,draw=color0,line width=1pt, fill=color0L,minimum width=0.2cm},
perceptron1/.style = {circle,draw=color1,line width=1pt,fill=color1L,minimum width=0.7cm},
perceptronS1/.style = {circle,draw=color1,line width=1pt,fill=color1L,minimum width=0.2cm},
perceptron2/.style = {circle,draw=color2,line width=1pt,fill=color2L,minimum width=0.7cm},
perceptronS2/.style = {circle,draw=color2,line width=1pt,fill=color2L,minimum width=0.2cm},
line0/.style = {draw=white,line width=3pt},
lineD/.style = {line width=1pt},
line1/.style = {draw=color1,line width=1pt},
line2/.style = {draw=color2,line width=1pt},
line3/.style = {draw=color3,line width=1pt},
operator0/.style = {draw, text=black, draw=black,line width=1pt,fill=color0L,minimum width=1cm,minimum size=1.5em},
operator1/.style = {draw, text=black, draw=color1,line width=1pt, fill=color1L,minimum width=1cm,minimum size=1.5em},
operator2/.style = {draw, text=black, draw=color2,line width=1pt, fill=color2L,minimum width=1cm,minimum size=1.5em},
brace0/.style = {decorate,decoration={brace,amplitude=5pt},color0},
dot/.style = {draw,fill,shape=circle,minimum size=5pt,inner sep=0pt},
dotwhite/.style = {draw,fill=white,shape=circle,minimum size=5pt,inner sep=0pt},
cross/.style={path picture={ \draw[thick,black](path picture bounding box.north) -- (path picture bounding box.south) (path picture bounding box.west) -- (path picture bounding box.east);	}},
circlewc/.style={draw,circle,cross,minimum width=0.3 cm},
dcross/.style={path picture={ \draw[thick, black](path picture bounding box.north west) -- (path picture bounding box.south east) (path picture bounding box.south west) -- (path picture bounding box.north east);}},
halfcross/.style={path picture={ \draw[thick, black](path picture bounding box.south west) -- (path picture bounding box.north east);}},
meter/.append style={draw, fill=white, inner sep=10, rectangle, font=\vphantom{A}, minimum width=30, line width=.7,
path picture={\draw[black] ([shift={(.1,.3)}]path picture bounding box.south west) to[bend left=50] ([shift={(-.1,.3)}]path picture bounding box.south east);\draw[black,-latex] ([shift={(0,.1)}]path picture bounding box.south) -- ([shift={(.3,-.1)}]path picture bounding box.north);}},
networkcircle1/.style = {line1,circle,fill=white, text width=2mm, minimum height=0.7cm},
networkellipse1/.style = {line1, ellipse,fill=white,minimum height=1.5cm,minimum width=0.7cm},
networkellipseM/.style = {line1, ellipse,fill=white,minimum height=2cm,minimum width=0.7cm},
networkellipseX/.style = {line1, ellipse,fill=white,minimum height=3cm,minimum width=0.9cm},
vertex0/.style = {regular polygon,regular polygon sides=7,draw, fill=color0L,minimum width=1cm},
vertex1/.style = {regular polygon,regular polygon sides=7,draw=color1,line width=1pt,fill=color1L,minimum width=1cm},
vertex2/.style = {regular polygon,regular polygon sides=7,draw=color2,line width=1pt,fill=color2L,minimum width=1cm},
}
\newcommand{\oneoneone}{\raisebox{-1.4pt}{\tikz[yscale=0.6,xscale=0.4]{
\node(1) [circle,draw,inner sep=0pt,minimum size=4.5pt] at (-1,0) {};
\node(2) [circle,draw,inner sep=0pt,minimum size=4.5pt] at (0,0) {};
\node(3) [circle,draw,inner sep=0pt,minimum size=4.5pt] at (1,0) {};
\node(e) [circle,draw=white,inner sep=0pt,minimum size=0.5pt] at (1.4,-0.2) {};
\draw (1)--(2) -- (3);
}}}
\newcommand{\oneothreeone}{\raisebox{-1.4pt}{\tikz[yscale=0.9,xscale=0.6]{
			\node(1) [circle,draw,inner sep=0pt,minimum size=4.5pt] at (-1,0.15) {};
			\node(3) [circle,draw,inner sep=0pt,minimum size=4.5pt] at (0,-.08) {};
			\node(4) [circle,draw,inner sep=0pt,minimum size=4.5pt] at (0,0.15) {};		
			\node(5) [circle,draw,inner sep=0pt,minimum size=4.5pt] at (0,0.38) {};			
			\node(6) [circle,draw,inner sep=0pt,minimum size=4.5pt] at (1,0.15) {};
			\node(e) [circle,draw=white,inner sep=0pt,minimum size=0.5pt] at (1.2,0.3) {};
			\draw (1)--(3);
			\draw (1)--(4);
			\draw (1)--(5);
			\draw (6)--(3);
			\draw (6)--(4);
			\draw (6)--(5);
}}}
\pgfplotsset{
/pgfplots/bar shift auto/.style={
/pgf/bar shift={%
-0.5*(\numplotsofactualtype/2*\pgfplotbarwidth + ((\numplotsofactualtype/2)-1)*(#1)) +
(.5+round((\plotnumofactualtype+1)/2)-1)*\pgfplotbarwidth + (round((\plotnumofactualtype+1)/2)-1)*(#1)
},
},
}

\makeatletter
\newcommand{\specificthanks}[1]{\@fnsymbol{#1}}
\makeatother

\begin{document}
\title{Dissipative quantum generative adversarial networks}
\author{Kerstin Beer \thanks{Institut f\"ur Theoretische Physik, Leibniz Universit\"at Hannover, Appelstraße 2, 30167 Hannover, Germany, kerstin.beer@itp.uni-hannover.de}
\and Gabriel M{\"u}ller \thanks{previously Institut f\"ur Theoretische Physik, presently Institut für Quantenoptik, Leibniz Universit\"at Hannover, Welfengarten 1, 30167 Hannover, Germany, g.mueller@iqo.uni-hannover.de}}

\maketitle

\textbf{Noisy intermediate-scale quantum (NISQ) devices build the first generation of quantum computers.
Quantum neural networks (QNNs) gained high interest as one of the few suitable quantum algorithms to run on these NISQ devices.
Most of the QNNs exploit supervised training algorithms with quantum states in form of pairs to learn their underlying relation.
However, only little attention has been given to unsupervised training algorithms despite interesting applications where the quantum data does not occur in pairs.
Here we propose an approach to unsupervised learning and reproducing characteristics of any given set of quantum states.
We build a generative adversarial model using two dissipative quantum neural networks (DQNNs), leading to the dissipative quantum generative adversarial network (DQGAN).
The generator DQNN aims to produce quantum states similar to the training data while the discriminator DQNN aims to distinguish the generator's output from the training data.
We find that training both parts in a competitive manner results in a well trained generative DQNN.
We see our contribution as a proof of concept for using DQGANs to learn and extend unlabeled training sets.}
\\\\

The last years brought forth the first generation of quantum computers, namely, noisy intermediate-scale quantum (NISQ) devices \cite{Preskill2018, Brooks2019, Arute2019}.
These devices are limited by high noise levels and a small number of qubits.
Additionally, the noise limits the quantum circuits to only short depth before the signal-to-noise ratio becomes too small.
Due to these limitations, today's quantum computers do not achieve fault-tolerant quantum computation \cite{Shor1996,Preskill1998}.
However, this is required to perform promising applications \cite{Nielsen2000} such as Shor's factoring algorithm \cite{Shor1994} or quantum simulation \cite{Lloyd1996}.

Quantum Neural Networks (QNNs) belong to the few quantum algorithms that can be executed on NISQ devices \cite{Cerezo2020}.
A QNN can be implemented as a hybrid quantum-classical algorithm \cite{Mitarai2018, Stokes2020, Schuld2019, Ostaszewski2019}.
It executes a short parameterised quantum circuit (PQC) \cite{Bu2021, Benedetti2019, Du2018} on a quantum computer while optimising its parameters classically.
This process saves important resources such that these QNNs can also work under the limitations of NISQ devices \cite{Beer2021a}.

Here, we focus on the dissipative QNN (DQNN) \cite{Beer2020} which features a unique set of qubits for each network layer.
It is universal for quantum computation and has achieved remarkable results for various applications \cite{Bondarenko2020, Beer2021}, including their application on actual NISQ devices \cite{Beer2021a}.
The DQNN propagates the input state's information through the network in a feed-forward manner.
This is realized through completely positive layer-to-layer transition maps that act on the qubits of two adjacent layers.
Each qubit is characterised by an individual quantum perceptron unitary, the fundamental building block of the DQNN.

Most of the current QNN applications are limited to supervised learning of labeled training states \cite{Cerezo2020}.
This includes their use for classification \cite{Schuld2019a, Schuld2020, Zhang2020}, denoising quantum data \cite{Bondarenko2020,Achache2020,Wan2017, Romero2017, Pepper2019}, learning unitary transformations \cite{Kiani2020,Geller2021, Verdon2018,Sedlak2019,Beer2020, Beer2021a} and learning graph-structured quantum data \cite{Verdon2018,Sedlak2019,Verdon2019,Cong2019,Arunachalam2017,Beer2021}.
In these problems, the QNNs are trained with respect to a special family of loss functions. These compare the network's output state for each input state to the respective target output state.
However, there are problems where these training pairs do not exist and a different approach is required for training. 

We introduce dissipative quantum generative adversarial networks (DQGANs) for unsupervised learning of an unlabeled training set.
The fundamental concept is adopted from the highly successful classical generative adversarial networks \cite{Goodfellow2014}. One DQGAN consists of two DQNNs, namely, the generator and the discriminator. The discriminator's aim is to distinguish real training states from fake states produced by the generator. The generator's aim is to produce fake states that the discriminator can not distinguish from the real ones. By alternately and adversarially training the networks the generator should learn the relevant features of the training set and produce states that extend the training set.

Our aim with the DQGAN is to extend a given unlabeled training set featuring any desired quantum states.
The current literature on generative adversarial learning with QNNs mainly aimed and achieved to reproduce labeled training states \cite{DallaireDemers2018,Niu2021}.
However, this neglects their major potential to produce states similar to a set of unlabeled training states without the need of any further information.
The latter usage enables a whole range of exciting possibilities.
For example, the DQGAN could be trained on a limited set of naturally produced quantum states and learn to produce similar quantum states for any imaginable further use in the future.

We realise the DQGAN in two ways, once by a simulation on a classical computer, another time by simulating a concrete quantum circuit implementation.
The first one allows a direct implementation of the quantum perceptron unitaries.
These can effectively be trained by the later introduced quantum backpropagation algorithm.
However, this implementation requires information that is not accessible during the execution on a NISQ device.
Therefore, we additionally provide a concrete quantum circuit implementation (DQGAN\textsubscript{Q}) which is trained by the gradient descent algorithm.

In contrast to some proposals of \emph{quantum generative adversarial networks} (QGANs) \cite{DallaireDemers2018,Lloyd2018a,Benedetti2019a,Chakrabarti2019,Hu2019,Zoufal2019, Zoufal2021,Niu2021,Stein2021, Huang2021} DQGAN is a fully quantum architecture and is trained with quantum data. In \cite{Zoufal2019, Zoufal2021}, for example, QGANs are defined as quantum-classical hybrid and include a quantum generator and a classical discriminator. Morover, the authors of \cite{Lloyd2018a} present the usage of quantum or classical data and two quantum processors in an adversarial learning setting. Furthermore, some of the QGAN proposals are trained with labelled data \cite{DallaireDemers2018}. In contrary the here proposed DQGAN can be trained with a set of unlabelled quantum states, to which only the discriminative model has access.  

In this work, we provide proof of concept for using DQGANs to learn and extend an unlabeled training set.
We train the DQGAN on two different toy models, a set of quantum states forming a line on the Bloch sphere and another set of clustered quantum states.
Here, we succeed on showing the vast potential of our DQGAN implementations.

Our paper is organized as follows: in \cref{section_QNN}, we describe the general concepts of our DQNN and follow up with the definitions of its two implementations. In \cref{section_QGAN}, we introduce the general concept of generative adversarial training with DQNNs, leading to the DQGAN and a concrete training algorithm. Conclusively, we present our results in \cref{section_results} and discuss their relevance to the field in \cref{section_discussion}.
\section{Quantum neural networks\label{section_QNN}}
Many attempts on builing a QNN, the quantum analogue of the popular classical neural network, have been made \cite{
Andrecut2002, 
Oliveira2008, 
Panella2011, 
Silva2016, 
Cao2017, 
Wan2017, 
Alvarez2017, 
Farhi2018, 
Killoran2019, 
Steinbrecher2019, 
Torrontegui2019, 
Sentis2019, 
Tacchino2020, 
Beer2020, 
Skolik2020, 
Zhang2020, 
Schuld2020, 
Sharma2020, 
Zhang2021 
}. In the following we describe the architecture of so-called \emph{dissipative quantum neural networks} (DQNNs) \cite{Beer2020, Beer2021, Beer2021a} as we will exploit this ansatz to form the DQGANs. We explain how their training algorithm can be simulated on a classical computer and how the DQNN can be implemented on a quantum computer \cite{Beer2021}.

\subsection{Dissipative quantum neural network\label{section_21}}
DQNNs are build of layers of qubits, which are connected via building blocks. Such a building block, named perceptron, is engineered as an arbitrary unitary operation.
We can express the propagation of a state $\rho^{\text{in}}$ trough the network as a composition of layer-to-layer transition maps, namely
\begin{equation}\label{eq:DQNN_rhoOut}
\rho^\text{out}=\mathcal{E}\left(\rho^{\text{in}}\right)= \mathcal{E}^{L+1}\left(\mathcal{E}^{L}\left(\dots \mathcal{E}^{2}\left(\mathcal{E}^{1}\left(\rho^{\text{in}}\right)\right)\dots\right)\right),
\end{equation}
where the transition maps are defined as
\begin{equation*}\label{eq:DQNN_E}
\mathcal{E}^{l}(X^{l-1}) \equiv \tr_{l-1}\big(\prod_{j=m_l}^{1} U^l_j (X^{l-1}\otimes \ket{0...0}_l\bra{ 0...0})\prod_{j=1}^{m_l} {U_j^l}^\dag\big),
\end{equation*}
where $U_j^l$ refers to the $j$th perceptron acting on layers $l-1$ and $l$, and $m_l$ is the total number of perceptrons connecting layers $l-1$ and $l$, see \cref{fig:DQNN_qnncircuitA}. These maps tensor the state of the current layer to the state of the next layer's qubits and apply the perceptron unitaries. Since the qubits from the first of the two layers are traced out additionally, these QNNs are called \emph{dissipative}.
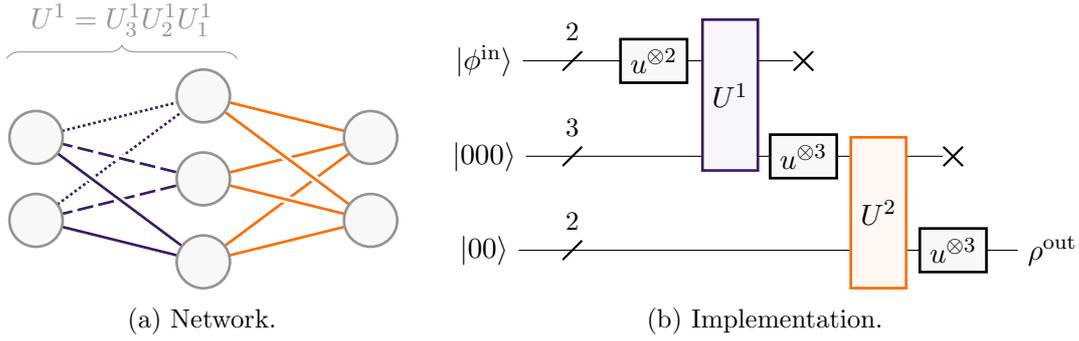
\begin{figure}
	\centering
	\begin{subfigure}[t]{0.35\linewidth}
		\centering
		\begin{tikzpicture}[scale=1.1]
			\begin{scope}[xshift=0.9cm,yshift=1.45cm]
				\draw[brace0] 
				(-1.25,0) -- node[above=1ex] {$U^1=U_3^1U_2^1U_1^1$}
				(1.5,0);   
			\end{scope}
			\draw[line0] (0,-.5) -- (2,1);
			\draw[line1,densely dotted] (0,-.5) -- (2,1);
			\draw[line0] (0,.5) -- (2,1);
			\draw[line1,densely dotted] (0,.5) -- (2,1);
			\draw[line0] (0,-.5) -- (2,0);
			\draw[line1, dash pattern=on 6pt off 2pt] (0,-.5) -- (2,0);
			\draw[line0] (0,.5) -- (2,0);
			\draw[line1, dash pattern=on 6pt off 2pt] (0,.5) -- (2,0);
			\draw[line0] (0,-.5) -- (2,-1);
			\draw[line1] (0,-.5) -- (2,-1);
			\draw[line0] (0,.5) -- (2,-1);
			\draw[line1] (0,.5) -- (2,-1);
			\foreach \x in {-1,0,1} {
				\draw[line0] (2,\x) -- (4,-0.5);
				\draw[line2] (2,\x) -- (4,-0.5);
				\draw[line0] (2,\x) -- (4,0.5);
				\draw[line2] (2,\x) -- (4,0.5);
			}
			\node[perceptron0] at (0,-0.5) {};
			\node[perceptron0] at (0,0.5) {};
			\node[perceptron0] at (2,-1) {};
			\node[perceptron0] at (2,0) {};
			\node[perceptron0] at (2,1) {};
			\node[perceptron0] at (4,-0.5) {};
			\node[perceptron0] at (4,0.5) {};
		\end{tikzpicture}
		\subcaption{Network. }
		\label{fig:DQNN_qnncircuitA}
	\end{subfigure}
	\begin{subfigure}[t]{0.64\linewidth}
		\centering
		\begin{tikzpicture}[scale=1.3]
			\matrix[row sep=0.3cm, column sep=0.4cm] (circuit) {
				\node(start3){$\ket{\phi^\text{in}}$};  
				& \node[halfcross,label={\small 2}] (c13){};
				& \node[operator0] (c23){$u^{\otimes 2}$};
				& \node[]{}; 
				& \node[dcross](end3){}; 
				& \node[]{}; 
				& \node[]{}; 
				& \node[]{}; \\
				\node(start2){$\ket{000}$};
				& \node[halfcross,label={\small 3}] (c12){};
				& \node[]{}; 
				& \node[]{}; 
				& \node[operator0] (c32){$u^{\otimes 3}$};
				& \node[]{}; 
				& \node[dcross](end2){}; 
				& \node[]{};  \\
				\node(start1){$\ket{00}$};
				& \node[halfcross,label={\small 2}] (c11){};
				& \node[]{}; 
				& \node[]{}; 
				& \node[]{}; 
				& \node[]{}; 
				& \node[operator0] (c41){$u^{\otimes 3}$};
				& \node (end1){$\rho ^\text{out}$}; \\
			};
			\begin{pgfonlayer}{background}
				\draw[] (start1) -- (end1)  
				(start2) -- (end2)
				(start3) -- (end3);
				\node[operator1, minimum height=2cm] at (-.35,0.5) {$U^1$};
				\node[operator2,minimum height=2cm] at (1.15,-.7) {$U^2$};
			\end{pgfonlayer}
		\end{tikzpicture}
		\subcaption{Implementation. }
		\label{fig:DQNN_qnncircuitB}
	\end{subfigure}
	\vspace*{10mm}
	\caption{\textbf{DQNN} An exemplary DQNN consisting of two layers of quantum perceptrons (a) can be implemented as quantum circuit (b). The $u$-gates represent layers of single qubit operations. $U^l=U_{m_l}^l \dots U_1^l$ denote the layer unitaries, where every unitary $U^l_k$ is expressed trough two-qubit unitaries, see \cite{Beer2020}.}
	\label{fig:DQNN_qnncircuit}
\end{figure}

The training of such an QNN architecture is done with respect to a data set containing $S$ input and desired output states, namely
$\{\ket{\phi^{\text{in}}_x}, \ket{\phi^{\text{SV}}_x}\} $.
For example, in \cite{Beer2020} it is shown that the DQNN algorithm can successfully characterize an unknown unitary $Y$, using desired output training states of the form $\ket{\phi^\text{SV}_x} = Y\ket{\phi^\text{in}_x}$.

Generally, the training is done via maximising a training loss function based on the fidelity $F$ of two states, e.g., of the form
\begin{equation}
\label{eq:DQNN_trainingloss}
\mathcal{L}_\text{SV}=\frac{1}{S}\sum_{x=1}^S F(\ket{\phi^{\text{SV}}_x}\bra{\phi^{\text{SV}}_x},\rho_x^{\text{out}}) = \frac{1}{S}\sum_{x=1}^S \braket{\phi^{\text{SV}}_x|\rho_x^{\text{out}}|\phi^{\text{SV}}_x}.
\end{equation}
The general aim is to optimise such a loss function by updating the variable parts of the DQNN.
In the following we explain the training process in two cases, the simulation on a classical computer and the quantum circuit implementation suitable for NISQ devices.

\subsection{Classical simulation implementation\label{section_QNN_cl}}
We can implement the algorithm using the quantum perceptron unitaries $U_j^l$ directly. Hence, every perceptron is described via $(2^{m_l+1})^2-1$ parameters.
Via feed-forward propagation of the input state trough the DQNN and back-propagation of the desired output state we can gain information on how every unitary $U_j^l$ has to be updated to minimize the training loss, exemplary defined in \cref{eq:DQNN_trainingloss}. We can formulate the update, using an update matrix $K_j^l(t)$, as
\begin{equation*}
\label{eq:DQNN_updateU}
U_j^l(t+\epsilon)=e^{i\epsilon K_j^l(t)} U_j^l(t),
\end{equation*}
where $\epsilon$ is the training step size and $t$ is the step parameter. The concrete formula of the update matrix is derived in \cite{Beer2020}. Remarkable is, that to evaluate the matrix $K_j^l(t)$, which updates a perceptron connecting layers $l-1$ and $l$, only two quantum states are needed: the output state of layer $l$ obtained by feed-forward propagation through the network, and the state of the layer $l+1$ obtained by back-propagation of the desired output. For more details of the classical simulation of a DQNN algorithm we point to \cite{Beer2020}. The code can be found at \cite{Github}.

\subsection{Quantum circuit implementation\label{section_QNN_q}}
To implement the quantum perceptrons on a quantum computer we have to abstract the perceptron unitaries into parameterised quantum gates. In \cite{Beer2021a} it is used that every arbitrary two-qubit unitary can be implemented with a two-qubit canonical gate and single qubit gates, see \cite{Zhang2003,Zhang2004,Blaauboer2008,Watts2013,Crooks2019, Peterson2020}. This yields the implementation of each perceptron via $m_{l-1}$ two-qubit unitaries connecting one qubit of the output layer $l$ to all qubits in the input layer $l-1$, respectively.

Rephrasing the sequence of single qubit gates in form of the gate $u$ and summarizing the two-qubit canonical gates in $U^l$ leads to the neat representation in \cref{fig:DQNN_qnncircuitB}. For the DQNN\textsubscript{Q}, $n = \sum_{l=1}^{L} m_l$ qubits are needed, where $L$ is the number of layers. The overall PQC consists of $3m + 3\sum_{l=1}^{L+1} m_{l-1}(1+m_l)$ parameters.

The DQNN\textsubscript{Q} implementation can be trained with gradient descent. At the beginning, the parameters of the quantum circuit are initialised as $\vec{\omega}_0$. All parameters are updated by $\vec{\omega}_{t+1} = \vec{\omega}_{t} + \vec{d \omega}_{t}$ in every training epoche, where $\vec{d\omega}_{t} = \eta {\nabla} \mathcal{L}_\text{SV} \left(\vec{\omega}_t\right)$ with the learning rate $\eta$ and the gradient is of the form
\begin{equation*}
\nabla _k \mathcal{L}_\text{SV} \left(\vec{\omega}_t\right) = \frac{\mathcal{L}_\text{SV}\left(\vec{\omega}_t + \vec\epsilon{e}_k\right) - \mathcal{L}_\text{SV}\left(\vec{\omega}_t - \vec\epsilon{e}_k\right)}{2\epsilon} + \mathcal{O}\left(\epsilon^2\right).
\end{equation*}
For a thorough explanation of the DQNN\textsubscript{Q} suitable for the execution on NISQ devices we refer to \cref{section:dqnn_q_implementation_details} and \cite{Beer2021a}.
\section{Dissipative quantum adversarial neural networks\label{section_QGAN}}
In the field of machine learning we can generally distinguish \emph{discriminative} and \emph{generative} models. For instance, classification problems such as classifying handwritten digits \cite{Nielsen2015} are common discriminative tasks. On the contrary, generative  models produce data. Speaking in the example of handwritten digits, we would train a generative model to produce different \enquote{handwritten} digits from random input.

In the following, we describe \emph{generative adversarial networks} (GANs). These are built of two models, where one of them has a generative and the other one a discriminative task. It is much harder to train generative models than discriminative models. The proposal of GANs offered new possibilities and has since found a lot of applications\cite{Creswell2018}, ranging from classification or regression tasks \cite{Creswell2018, Zhu2016, Salimans2016} to the generation \cite{Reed2016} and improvement \cite{Ledig2017} of images. 

\subsection{General concept\label{section_31}}
GANs were first introduced in \cite{Goodfellow2014}. The generative and discriminative parts of their GAN are implemented as a multi-layer perceptron, respectively. One the one hand, the generative model gets random noise samples as input and produces synthetic samples. On the other hand, the discriminator has access to both the generator's output and samples from the training data. In the original proposal, this data cannot be accessed by the generator.

The training aim of the discriminator is to distinguish correctly between the training data and the synthetic data. Since the generator's goal is to trick the discriminative model, the problem is called a \emph{minimax problem}. 
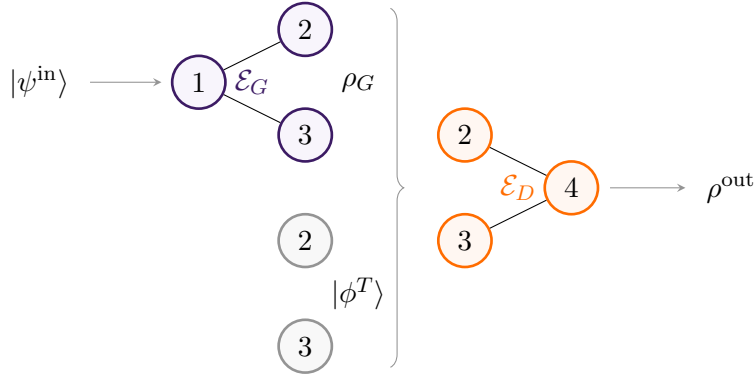
\begin{figure}[h!]
\centering
\begin{tikzpicture}[scale=1.4]
\node[] (pz) at (-1.5,0) {$\ket{\psi^\text{in}}$};
\node[perceptron1] (a) at (0,0) {1};
\node[perceptron1] (b) at (1,-.5) {3};
\node[perceptron1] (c) at (1,.5) {2};
\draw (a) -- (b);
\draw (a) -- (c);
\node[color1] (Ug) at (.5,0) {$\mathcal{E}_G$};
\node[] (Rg) at (1.5,0) {$\rho_G$};
\draw[-stealth,shorten <=4pt, shorten >=4pt,color0] (pz) -- (-.25,0);
\draw[brace0](1.8,0.7)-- (1.8,-2.7) ;
\begin{scope}[shift={(0,-2)}]
	\node[] (Rt) at (1.5,0) {$\ket{\phi^T}$};
	\node[perceptron0] (b) at (1,-.5) {3};
	\node[perceptron0] (c) at (1,.5) {2};
\end{scope}
\begin{scope}[shift={(2.5,-1)}]
	\node[] (pd) at (2.5,0) {$\rho^\text{out}$};
	\node[perceptron2] (d) at (0,-.5) {3};
	\node[perceptron2] (e) at (0,.5) {2};
	\node[perceptron2] (f) at (1,0) {4};
	\draw (d) -- (f);
	\draw (e) -- (f);
	\node[color2] (Ud) at (.5,0) {$\mathcal{E}_D$};
	\draw[-stealth,shorten <=4pt, shorten >=4pt,color0] (f) -- (pd);
\end{scope}
\end{tikzpicture}
\caption{\textbf{DQGAN.} The depicted DQGAN consists of four qubits. Qubits $2$ and $3$ are shared by the generative and the discriminative QNN. The state of this qubits is either the generator's output state $\rho_G$ on the input state, i.e., $\rho_G=\mathcal{E}_G (\ket{\psi^\text{in}}\bra{\psi^\text{in}})$ or a given training state $\ket{\phi^T}$. }
\label{fig:QGAN_qgan}
\end{figure}

Following the above-described ansatz, the DQGAN is constructed of two DQNNs, the generative model, and the discriminative model, described through the completely positive maps $\mathcal{E}_G$ and $\mathcal{E}_D$, respectively. The number of qubits in the generator's last layer equals the number of qubits in the discriminator's first layer. Hence, the generator's output can be used as input for the discriminator.

For the training a set of training states $\{\ket{\phi_x^T}\}_{x=1}^N$ and a set of random states $\{\ket{\psi ^\text{in}_x}\}$ is used. We assume the the states of both sets to be pure. The overall goal is to adversarially train both DQNNs, so that the generator produces states with characteristics similar to the training data.
We can describe the output of the discriminator DQNN as
\begin{singlespace}
\begin{equation*}
\rho^\text{out}=
\begin{cases} 
	\mathcal{E}_D (\mathcal{E}_G (\ket{\psi^\text{in}}\bra{\psi^\text{in}})) &\mbox{for generated data}\\
	\mathcal{E}_D (\ket{\phi^T}\bra{\phi^T}) & \mbox{for training data.} 
\end{cases}
\end{equation*}
\end{singlespace}

To be more precise we shortly discuss the DQGAN depicted in \cref{fig:QGAN_qgan} and consisting of four qubits. Please consider \cref{fig:DQNN_qnncircuit} for a better understanding of the following description. The generative model consists of two two-qubit unitaries $U_{G1}$ and $U_{G2}$, acting on qubits $1$ and $2$, and qubits $1$ and $3$, respectively. The discriminator is described by a single three-qubit unitary $U_D$. If the discriminative model gets a training data state $\ket{\phi^T}$ as input the resulting discriminator output state can be described as
\begin{align*}
\rho_{\mathrm{out}}^{D}
=&\tr_{\{2,3\}}\Big(U_D \left( \ket{\phi^T}\bra{\phi^T} \otimes \ket{0}\bra{0} \right) U_D^\dagger \Big).
\end{align*}
For the the generator's output as input, the discriminator has the output
\begin{align*}
\rho_{\mathrm{out}}^{G+D}
=&\tr_{\{1,2,3\}}\Big(U_D  U_{G2} U_{G1} ( \ket{\psi_x^\text{in}}\bra{\psi_x^\text{in}} \otimes \ket{000}\bra{000}) U_{G1}^\dagger U_{G2}^\dagger U_D^\dagger \Big).
\end{align*}
The general form of these output states is used in the proof of \cref{prop:QGAN_K}. 

The original DQNN approach focuses on characterising a relation between input and output data pairs. However, we try to characterise a data set of single quantum states instead. We aim to train a generative model in a way that it is able to produce quantum states with similar properties compared to the training data set. Such extended quantum data sets can be, for example, useful for experiments or training other QNN architectures. 

\subsection{Training algorithm\label{section_32}}
In analogy to the classical case described in \cite{Goodfellow2014} we can describe the training process through
\begin{equation}\label{eq:dqgan_minimax}
\mathcolorbox{\min_G 
\max_D \left(\frac{1}{S}\sum_{x=1}^S \bra{0} \mathcal{E}_D (\mathcal{E}_G (\ket{\psi_x^\text{in}}\bra{\psi_x^\text{in}}))\ket{0} + \frac{1}{S}\sum_{x=1}^S \bra{1} \mathcal{E}_D (\ket{\phi_x^T}\bra{\phi_x^T})\ket{1} \right).
}
\end{equation}
The updates of the discriminator and the generator take place alternately. For updating the generator we maximise the loss function
\begin{equation*}
\mathcal{L}_{D}(\mathcal{E}_D,\mathcal{E}_G)=\frac{1}{S}\sum_{x=1}^S \bra{0} \mathcal{E}_D (\mathcal{E}_G (\ket{\psi_x^\text{in}}\bra{\psi_x^\text{in}}))\ket{0} + \frac{1}{S}\sum_{x=1}^S \bra{1} \mathcal{E}_D (\ket{\phi_x^T}\bra{\phi_x^T})\ket{1}
\end{equation*}
for $r_D$ rounds, whereas the generator is trained through maximising
\begin{equation*}
\mathcal{L}_{G}(\mathcal{E}_D,\mathcal{E}_G)=\frac{1}{S}\sum_{x=1}^S \bra{1} \mathcal{E}_D (\mathcal{E}_G (\ket{\psi_x^\text{in}}\bra{\psi_x^\text{in}}))\ket{1}.
\end{equation*}
assume for $r_G$ rounds. Note that $\mathcal{L}_G$ differs from the corresponding term in \cref{eq:dqgan_minimax} in that the fidelity is calculated with respect to $\ket{1}$ instead of $\ket{0}$. Therefore, the generator is trained by maximising $\mathcal{L}_G$ rather than minimising. These procedures are repeated for $r_T$ epochs. The overall training algorithm is described in \cref{alg:QGAN_algorithmQ}.

\begin{algorithm}[H]
\caption{Training of the DQGAN.}
\label{alg:QGAN_algorithmQ}
\begin{algorithmic}
\State initialize network unitaries 
\For{$r_T$ epochs}
\State make a list of $S$ randomly chosen states of the training data list $\{\ket{\phi_x^T}\}_{x=1}^N$
\For{$r_D$ epochs}
\State make a list of $S$ random states $\ket{\psi _x^\text{in}}$
\State update the discriminator unitaries by maximizing $\mathcal{L}_{D}$
\EndFor
\For{$r_G$ epochs}
\State make a list of $S$ random states $\ket{\psi_x^\text{in}}$
\State update the generator unitaries by maximising $\mathcal{L}_{G}$
\EndFor
\EndFor
\State make a list of $V$ random states $\ket{\psi_x^\text{in}}$
\State propagate each $\ket{\psi_x^\text{in}}$ through the generator to produce $V$ new states
\State calculate $\mathcal{L}_{V}$
\end{algorithmic}
\end{algorithm}

\begin{repprop}{prop:QGAN_K}
The update matrix for a QGAN trained with pure states $\ket{\phi^T_x}$  has to be of the form
\begin{equation*}
K^l_j(t) = \frac{\eta 2^{m_{l-1}}i}{S}\sum_x\tr_\text{rest}\big(M^l_{j}(x,t)\big),
\end{equation*}
where 
\begin{align*}
M_j^l =& \Big[ U_{j}^{l} \dots U_{1}^{1} \ ( \ket{\psi_x^\text{in}}\bra{\psi_x^\text{in}} \otimes \ket{0...0}\bra{0...0})  U_{1}^{1 \dagger} \dots U_{j}^{l \dagger}, \\
&U_{j+1}^{l\dagger}\dots U_{m_{L+1}}^{L+1 \dagger} \left(\mathbbm{1}_\mathrm{in+hid}\otimes \ket{1}\bra{1}\right)U_{m_{L+1}}^{L+1 } \dots U_{l+1}^{l}\Big]
\end{align*}
for $l\le g$ and 
\begin{align*}
M_j^l =& \Big[  U_{j}^{l} \dots U_{1}^{g+1} \ket{\phi^T}\bra{\phi^T} \otimes \ket{0...0}\bra{0...0} U_{1}^{g+1 \dagger}\dots U_{j}^{l \dagger}   \\
&- U_{j}^{l} \dots U_{1}^{g+1} U_{m_g}^{g} \dots U_{1}^{1} \ ( \ket{\psi_x^\text{in}}\bra{\psi_x^\text{in}} \otimes \ket{0...0}\bra{0...0})   U_{1}^{1 \dagger} \dots U_{m_{g}}^{g\dagger}  U_{1}^{g+1 \dagger}\dots U_{j}^{l \dagger} ,\\
&U_{j+1}^{l\dagger}\dots U_{m_{L+1}}^{L+1 \dagger} \left(\mathbbm{1}_\mathrm{in+hid}\otimes \ket{1}\bra{1}\right)U_{m_{L+1}}^{L+1 } \dots U_{l+1}^{l}\Big]
\end{align*}
else. Here, $U_j^l$ is assigned to the $j$th perceptron acting on layers $l-1$and $l$, $g$ is the number of perceptron layers of the generator, and $\eta$ is the learning rate.
\end{repprop}
The proof can be found in \cref{section_derivation}. Note, that in the following only DQGANs of three layers are used, i.e.\ both DQNNs are built of two qubit layers connected by one perceptron layer, respectively. Hence, we assume $g=1$ in the following.

In analogy to training the DQNN\textsubscript{Q}, the implementation on a quantum computer is done via parameterised quantum gates, which are updated using gradient descent. The training losses $\mathcal{L}_G$ and $\mathcal{L}_D$ are evaluated via measurement of the discriminator's output qubit.

At the end of the training the goal is that every generator output is close to at least one of the given states $\{\ket{\phi_x^T}\}_{x=1}^N$. To test this we additionally generate $V$ random states $\ket{\psi^\text{in}}$ as input states of the generator. We refer to the corresponding generated states as validation states. For each validation state, we search for the closest state of the data set via $\max_{x=1}^N \left( \bra{\phi_x^T} \mathcal{E}_G (\ket{\psi_i^\text{in}}\bra{\psi_i^\text{in}})\ket{\phi_x^T}\right)$. Using all validation states we define the \emph{validation loss}
\begin{equation*}
\mathcal{L}_{V}(\mathcal{E}_G)=\frac{1}{V}\sum_{i=1}^V \max_{x=1}^N \left( \bra{\phi_x^T} \mathcal{E}_G (\ket{\psi_i^\text{in}}\bra{\psi_i^\text{in}})\ket{\phi_x^T}\right).
\end{equation*}

Note that the above-defined validation loss would be optimised indeed if the generator produces only a small variety of states or even exactly one state. As long as these are close to at least one of the training states, the validation loss is high. Therefore, it is important to check the diversity of the generator's output, which will be described in \cref{section_discussion}.

\section{Results\label{section_results}}

In the following we test the training algorithm including the two training functions $\mathcal{L}_{G}$ and $\mathcal{L}_{D}$. Here, we use the simulation on a classical computer. The code can be found at \cite{Github}. As the training data we prepare a set of pure one-qubit states which build a line on the Bloch sphere, namely
\begin{equation*}
\text{data}_\text{line}=\left\{\frac{(N-x)\ket{0}+(x-1)\ket{1}}{||(N-x)\ket{0}+(x-1)\ket{1}||}\right\}_{x=1}^{N},
\end{equation*}
for $N=50$.
Next, we randomly shuffle this set of states. The first $S$ of the resulting set $\{\ket{\psi^T_x}\}_{x=1}^{S}$ will be used for the training process. The full data set $\{\ket{\psi^T_x}\}_{x=1}^{N}$ is used for computing the validation loss. 

In \cref{fig:GAN_line} the evolution of the discriminator's and generator's training losses and the validation loss is plotted. The latter reaches values over $0.95$ at $t=9.5$, i.e., after training round $r_T=475$. Moreover, we can observe that in the first training epochs, the training loss of the generator shrinks and the discriminator training loss increases. This behaviour inverts at $t\approx2$. For the remaining training process, this switch between an increasing generator training loss and an increasing discriminator training loss happens repetitively. We explain this behaviour with the opposing goals of the generator and the discriminator and a changing dominance of one of the networks. 

\afterpage{%
\thispagestyle{empty}
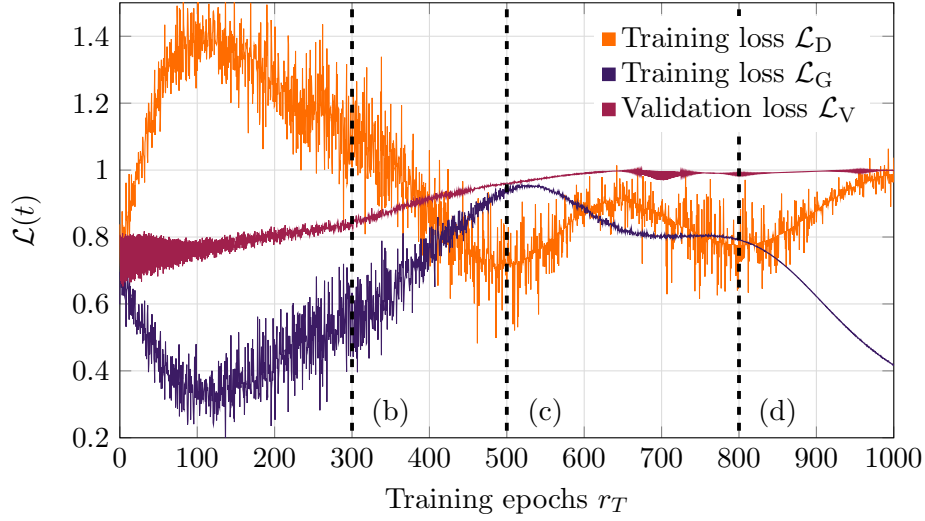
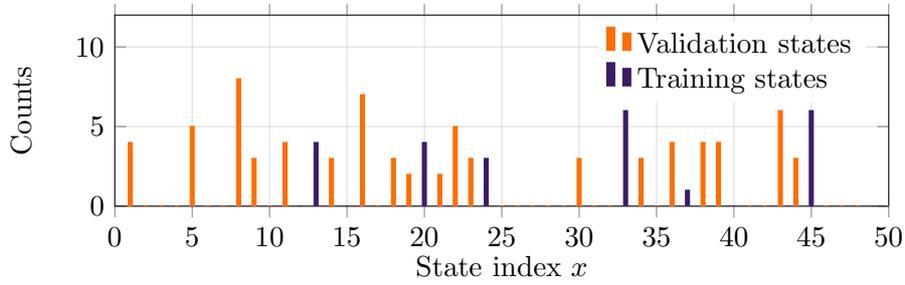
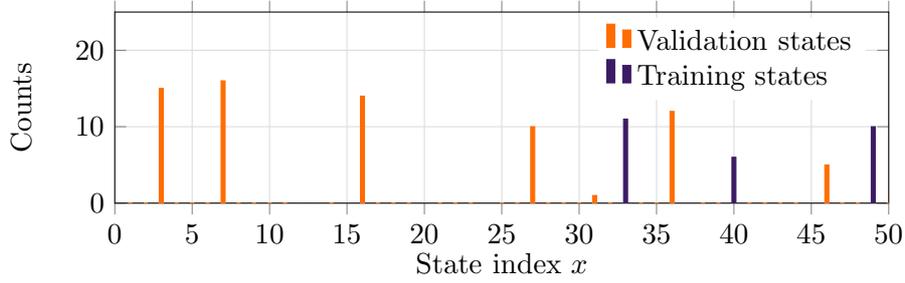
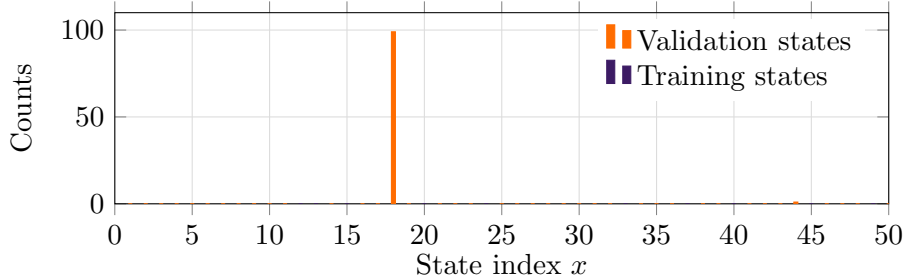
\begin{figure}[H]
\centering
\begin{subfigure}{\textwidth}\centering
\begin{tikzpicture}
\begin{axis}[
xmin=0,   xmax=20,
ymin=0.2,   ymax=1.5,
width=0.8\linewidth, 
height=0.5\linewidth,
grid=major,grid style={color0M},
xlabel= Training epochs $r_T$, 
xticklabels={-100,0,100,200,300,400,500,600,700,800,900,1000},
ylabel=$\mathcal{L}(t)$,legend pos=north east,legend cell align={left},legend style={draw=none,legend image code/.code={\filldraw[##1] (-.5ex,-.5ex) rectangle (0.5ex,0.5ex);}}]
\coordinate (0,0) ;
\addplot[mark size=1.5 pt,  color=color2] table [x=step times epsilon, y=costFunctionDis, col sep=comma] {numerics/QGAN_50data10sv_100statData_100statData_1-1networkGen_1-1networkDis_lda1_ep0i01_rounds1000_roundsGen1_roundsDis1_line_plot1_training.csv};
\addlegendentry[scale=1]{Training loss $\mathcal{L}_\text{D}$} 
\addplot[mark size=1.5 pt,  color=color1] table [x=step times epsilon, y=costFunctionGen, col sep=comma] {numerics/QGAN_50data10sv_100statData_100statData_1-1networkGen_1-1networkDis_lda1_ep0i01_rounds1000_roundsGen1_roundsDis1_line_plot1_training.csv};
\addlegendentry[scale=1]{Training loss $\mathcal{L}_\text{G}$} 
\addplot[mark size=1.5 pt,  color=color3] table [x=step times epsilon, y=costFunctionTest, col sep=comma] {numerics/QGAN_50data10sv_100statData_100statData_1-1networkGen_1-1networkDis_lda1_ep0i01_rounds1000_roundsGen1_roundsDis1_line_plot1_training.csv};
\addlegendentry[scale=1]{Validation loss $\mathcal{L}_\text{V}$} 
\draw [line width=0.5mm,dashed] (60,0) -- (60,200);
\node at (70,10) {(b)};
\draw [line width=0.5mm,dashed] (100,0) -- (100,200);
\node at (110,10) {(c)};
\draw [line width=0.5mm,dashed] (160,0) -- (160,200);
\node at (170,10) {(d)};
\end{axis}
\end{tikzpicture}
\subcaption{Loss functions.}
\label{fig:GAN_line}
\end{subfigure}
\begin{subfigure}{\textwidth}\centering
\begin{tikzpicture}[scale=1]
\begin{axis}[
ybar,
bar width=1.5pt,
xmin=0,   xmax=50,
ymin=0,   ymax=12,
width=.8\linewidth, 
height=.28\linewidth,
grid=major,
grid style={color0M},
xlabel= State index $x$, 
ylabel=Counts,legend pos=north east,legend cell align={left},legend style={draw=none,legend image code/.code={\filldraw[##1] (-.5ex,-.5ex) rectangle (0.5ex,0.5ex);}}]
\addplot[color=color2, fill=color2] table [x=indexDataTest, y=countOutTest, col sep=comma] {numerics/QGAN_50data10sv_100statData_100statData_1-1networkGen_1-1networkDis_lda1_ep0i01_rounds300_roundsGen1_roundsDis1_line_plot1_statisticsUSV.csv};
\addlegendentry[scale=1]{Validation states} 
\addplot[color=color1,fill=color1] table [x=indexDataTrain, y=countOutTrain, col sep=comma] {numerics/QGAN_50data10sv_100statData_100statData_1-1networkGen_1-1networkDis_lda1_ep0i01_rounds300_roundsGen1_roundsDis1_line_plot1_statisticsSV.csv};
\addlegendentry[scale=1]{Training states} 
\end{axis}
\end{tikzpicture}
\subcaption{Diversity of the generator's output ater $r_T=300$ training epochs.} \label{fig:GAN_line300}
\end{subfigure}
\begin{subfigure}{\textwidth}\centering
\begin{tikzpicture}[scale=1]
\begin{axis}[
ybar,
bar width=1.5pt,
xmin=0,   xmax=50,
ymin=0,   ymax=25,
width=.8\linewidth, 
height=.28\linewidth,
grid=major,
grid style={color0M},
xlabel= State index $x$, 
ylabel=Counts,legend pos=north east,legend cell align={left},legend style={draw=none,legend image code/.code={\filldraw[##1] (-.5ex,-.5ex) rectangle (0.5ex,0.5ex);}}]
\addplot[color=color2, fill=color2] table [x=indexDataTest, y=countOutTest, col sep=comma] {numerics/QGAN_50data10sv_100statData_100statData_1-1networkGen_1-1networkDis_lda1_ep0i01_rounds500_roundsGen1_roundsDis1_line_plot1_statisticsUSV.csv};
\addlegendentry[scale=1]{Validation states} 
\addplot[color=color1,fill=color1] table [x=indexDataTrain, y=countOutTrain, col sep=comma] {numerics/QGAN_50data10sv_100statData_100statData_1-1networkGen_1-1networkDis_lda1_ep0i01_rounds500_roundsGen1_roundsDis1_line_plot1_statisticsSV.csv};
\addlegendentry[scale=1]{Training states} 
\end{axis}
\end{tikzpicture}
\subcaption{Diversity of the generator's output ater $r_T=500$ training epochs.} \label{fig:GAN_line500}
\end{subfigure}
\begin{subfigure}{\textwidth}\centering
\begin{tikzpicture}[scale=1]
\begin{axis}[
ybar,
bar width=1.5pt,
xmin=0,   xmax=50,
ymin=0,   ymax=110,
width=.8\linewidth, 
height=.28\linewidth,
grid=major,
grid style={color0M},
xlabel= State index $x$, 
ylabel=Counts,legend pos=north east,legend cell align={left},legend style={draw=none,legend image code/.code={\filldraw[##1] (-.5ex,-.5ex) rectangle (0.5ex,0.5ex);}}]
\addplot[color=color2, fill=color2] table [x=indexDataTest, y=countOutTest, col sep=comma] {numerics/QGAN_50data10sv_100statData_100statData_1-1networkGen_1-1networkDis_lda1_ep0i01_rounds800_roundsGen1_roundsDis1_line_plot1_statisticsUSV.csv};
\addlegendentry[scale=1]{Validation states} 
\addplot[color=color1,fill=color1] table [x=indexDataTrain, y=countOutTrain, col sep=comma] {numerics/QGAN_50data10sv_100statData_100statData_1-1networkGen_1-1networkDis_lda1_ep0i01_rounds800_roundsGen1_roundsDis1_line_plot1_statisticsSV.csv};
\addlegendentry[scale=1]{Training states} 
\end{axis}
\end{tikzpicture}
\subcaption{Diversity of the generator's output ater $r_T=800$ training epochs.} \label{fig:GAN_line800}
\end{subfigure}
\caption{\textbf{Training a DQGAN.} (a) depicts the evolution of the loss functions during the training of a \protect\oneoneone DQGAN in $r_T=1000$ epochs with $\eta=1$ and $\epsilon=0.01$ using $50$ data pairs where $10$ are used as training states. The dashed lines mark the diversity checks 300 (b), 500 (c) and 800 (d) for the generator's output.}
\end{figure}}

%
The saturating validation loss in \cref{fig:GAN_line} gives the impression that the longer we train the DQGAN, the better the results. In the original proposal of the DQNN \cite{Beer2020} this was the case. However, the validation loss would be maximal if the generator would permanently produce the exact same state when this state is one of the training states $\{\ket{\psi^T_x}\}_{x=1}^{N}$. This would not fit our aim to train the generator to extended the training set. Hence, we explain in the following how we check the \emph{diversity} of the generator's output. 

After training for $r_T$ rounds, we use the generator to produce a set of $100$ states. Using the fidelity, we find for each of these states the element of $\text{data}_\text{line}$, which is the closest. In this way, we obtain a number for every index $x$ of this set describing how often the output of the generator was most closely to the $x$th element of $\text{data}_\text{line}$. In \cref{fig:GAN_line300,fig:GAN_line500,fig:GAN_line800} we plot these numbers in the form of an histogram. The different colours describe whether an element of $\text{data}_\text{line}$ was used as a training state or not. We find that the diversity was good after $300$ training epochs. However, it decreases afterwards in the ongoing training. We point to \cref{apnx:numerics} for more numerical results.

\begin{figure}[H]
	\centering
	\begin{subfigure}{\textwidth}\centering
		\begin{tikzpicture}[scale=1]
			\begin{axis}[
				ybar,
				bar width=1.5pt,
				xmin=0,   xmax=50,
				ymin=0,   ymax=20,
				width=.8\linewidth, 
				height=.28\linewidth,
				grid=major,
				grid style={color0M},
				xlabel= State index $x$, 
				ylabel=Counts,legend pos=north east,legend cell align={left},legend style={draw=none,legend image code/.code={\filldraw[##1] (-.5ex,-.5ex) rectangle (0.5ex,0.5ex);}}]
				\addplot[color=color2, fill=color2] table [x=indexDataTest, y=countOutTest, col sep=comma] {numerics/dqnn_q_eq_line_v2_epoch_100_vs.csv};
				\addlegendentry[scale=1]{Validation states} 
				\addplot[color=color1,fill=color1] table [x=indexDataTrain, y=countOutTrain, col sep=comma] {numerics/dqnn_q_eq_line_v2_epoch_100_ts.csv};
				\addlegendentry[scale=1]{Training states} 
			\end{axis}
		\end{tikzpicture}
		\subcaption{Diversity of the generator's output ater $r_T=100$ training epochs.} \label{fig:dqnn_q_eq_line_a}
	\end{subfigure}
	\begin{subfigure}{\textwidth}\centering
		\begin{tikzpicture}[scale=1]
			\begin{axis}[
				ybar,
				bar width=1.5pt,
				xmin=0,   xmax=50,
				ymin=0,   ymax=20,
				width=.8\linewidth, 
				height=.28\linewidth,
				grid=major,
				grid style={color0M},
				xlabel= State index $x$, 
				ylabel=Counts,legend pos=north east,legend cell align={left},legend style={draw=none,legend image code/.code={\filldraw[##1] (-.5ex,-.5ex) rectangle (0.5ex,0.5ex);}}]
				\addplot[color=color2, fill=color2] table [x=indexDataTest, y=countOutTest, col sep=comma] {numerics/dqnn_q_eq_line_v2_epoch_440_vs.csv};
				\addlegendentry[scale=1]{Validation states} 
				\addplot[color=color1,fill=color1] table [x=indexDataTrain, y=countOutTrain, col sep=comma] {numerics/dqnn_q_eq_line_v2_epoch_440_ts.csv};
				\addlegendentry[scale=1]{Training states} 
			\end{axis}
		\end{tikzpicture}
		\subcaption{Diversity of the generator's output ater $r_T=440$ training epochs.} \label{fig:dqnn_q_eq_line_b}
	\end{subfigure}
	\caption{\textbf{Training a DQGAN\textsubscript{Q}.} The training set features $S=10$ equally spaced quantum states from $\text{data}_\text{line}$. The remaining states from $\text{data}_\text{line}$ are used as validation states. The DQGAN\textsubscript{Q} features a 1-1$^+$ generator DQNN\textsubscript{Q} and a 1-1$^+$ discriminator DQNN\textsubscript{Q}, and employs the hyper-parameters $r_D=4$, $\eta_D=0.5$, $r_G=1$ and $\eta_G=0.1$.}
	\label{fig:dqnn_q_eq_line}
\end{figure}
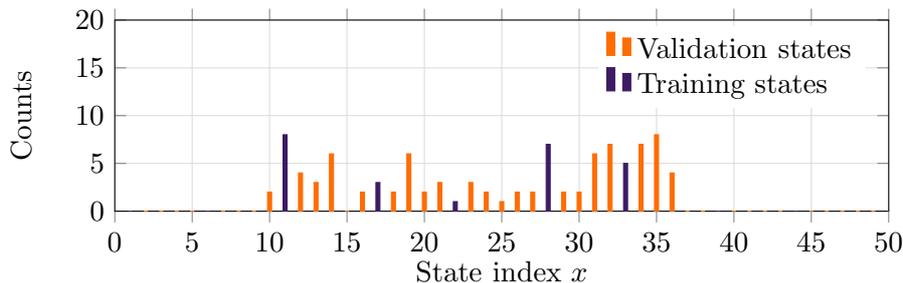
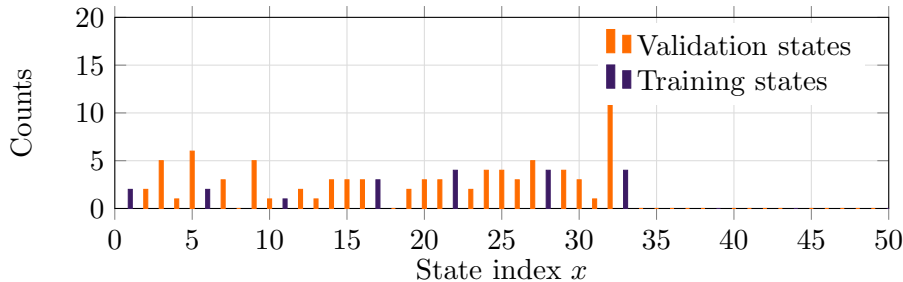

In addition to the DQGAN simulation on a classical computer we also simulate the DQGAN\textsubscript{Q} under noiseless circumstances. Here, the same training loss functions $\mathcal{L}_G, \mathcal{L}_D$ are used as well as the same training data, $\text{data}_\text{line}$. In this case, the training states are not picked randomly but $S=10$ equally spaced training states are chosen from $\text{data}_\text{line}$. The hyper-parameters of the training are chosen such that for each of the $r_T$ epochs, a 1-1$^+$ discriminator DQNN\textsubscript{Q} is trained $r_D=4$ times with a learning rate $\eta_D=0.5$ and a 1-1$^+$ generator is trained $r_G=1$ times with a learning rate $\eta_G=0.1$. The $+$ denotes a slightly different implementation of DQGAN\textsubscript{Q} compared to implementation discussed in \cite{Beer2021a}. It uses additional gates and is explained in \cref{alg:QGAN_algorithmQ}.

The results of training the DQGAN\textsubscript{Q} are shown in are shown in \cref{fig:dqnn_q_eq_line}. The generator's diversity after training for $r_T=100$ epochs is depicted in \cref{fig:dqnn_q_eq_line_a}. Here, the generator achieves to produce states in a little more than half of the training data range. After $r_T=440$ training epochs, the generator's diversity is improved to two-thirds of the training data range as depicted in \cref{fig:dqnn_q_eq_line_b}. Please note that in both cases, the majority of the generator's produced states is closer to a validation state than a training state. This can be seen as a training success as the generator does not only learn to reproduce the training states but instead learns to extend the given training data.

For more numerical results using DQGAN\textsubscript{Q} we point to \cref{apnx:numerics} and \cite{Mueller2021}.
\section{Discussion\label{section_discussion}}
In this work, we introduced DQGANs, generative adversarial models based on the DQNN proposed in \cite{Beer2020}. A DQGAN features two DQNNs trained in an opposing manner: the discriminator's goal is to distinguish between synthetic, by the generator produced quantum states and elements of the training data set. On the contrary, the generator aims to produce data with similar properties as the states included in the training data set.

We aimed to extend a given data set with states with similar characteristics. Our examples have shown that this goal can be reached when training a DQGAN.
However, due to limitations in computational power, we could only train small DQGAN architectures and therefore leave questions open for future research. It would be interesting if using larger DQNNs for the generator or the discriminator leads to better validation loss values or more diversity in the generator's output. 
Further, the study of other data sets is of interest. One example could be a set of states with similar degrees of entanglement (with respect to a chosen entanglement measure) \cite{Schatzki2021}.
Since in classical machine learning, the output of GANs is often used to train other neural network architectures, a similar application for DQGANs and DQNNs is conceivable. 

\paragraph{Acknowledgements}
The authors would like to thank Tobias J. Osborne and Christian Struckmann for valuable discussions. Moreover, helpful correspondence with Dmytro Bondarenko, Terry Farrelly, Polina Feldmann, Daniel List, Jan Hendrik Pfau, Robert Salzmann, Daniel Scheiermann, Viktoria Schmiesing, Marvin Schwiering, and Ramona Wolf is gratefully acknowledged. This work was supported, in part, by the DFG through SFB 1227 (DQ-mat), Quantum Valley Lower Saxony, and funded by the Deutsche Forschungsgemeinschaft (DFG, German Research Foundation) under Germanys Excellence Strategy EXC-2123 QuantumFrontiers 390837967.
\newpage
\bibliographystyle{nosort_habbrv} 
\bibliography{literatur}
\newpage
\appendix
\section{Derivation of the update matrices}\label{section_derivation}

Analogously to the DQNN update rule presented in \cite{Beer2020} the unitaries will be updated through
\begin{equation*}
	U_j^l(t+\epsilon)=e^{i\epsilon K_{j}^l(t)} U_j^l(t).
\end{equation*}
We will derive the update matrices in general in \cref{prop:QGAN_K}. To understand the basic idea we discuss the update of a DQGAN consisting of three unitaries, see \cref{{fig:QGAN_qgan}}, first. These perceptron unitaries have the following update rules:
\begin{align*}
	U_D(t+\epsilon)&=e^{i\epsilon K_{D}(t)} U_D(t)\\ 
	U_{G1}(t+\epsilon)&=e^{i\epsilon K_{G1}(t)} U_{G1}(t)\\  
	U_{G2}(t+\epsilon)&=e^{i\epsilon K_{G2}(t)} U_{G2}(t).  
\end{align*}
Note  that the unitaries act on the current layers, e.g. is $U_{G1}$ denotes $U_{G1}\otimes \mathbbm{1}$ and $U_{G2}$ denotes $\mathbbm{1} \otimes U_{G2}$. 

In the first part of the algorithm the generator is fixed and only the discriminator is updated. When the training data is the discriminator's input we get the output state 
\begin{align*}
	\rho_{\mathrm{out}}^{D}(t+\epsilon)
	=&\tr_{\{2,3\}}\Big(e^{i\epsilon K_{D}}U_D \ \ket{\phi^T}\bra{\phi^T} \otimes \ket{0}\bra{0} \ U_D^\dagger e^{-i\epsilon K_{D}}\Big)\\
	=&\tr_{\{2,3\}}\Big(U_D\ \ket{\phi^T}\bra{\phi^T} \otimes \ket{0}\bra{0}\ U_D^\dagger +i\epsilon\ \left[K_{D}, U_D\ \ket{\phi^T}\bra{\phi^T} \otimes \ket{0}\bra{0}U_D^\dagger\right]  \\
	&+\mathcal{O}(\epsilon^2)\Big)\\	
	=&\rho_{\mathrm{out}}^{D}(t)+i\epsilon\ \tr_{\{2,3\}}\Big(\left[K_{D}, U_D\ \ket{\phi^T}\bra{\phi^T} \otimes \ket{0}\bra{0} \ U_D^\dagger\right]  \Big)+\mathcal{O}(\epsilon^2).
\end{align*}
If the discriminator gets the generator's output as input we get the output state
\begin{align*}
	\rho_{\mathrm{out}}^{G+D}(t+\epsilon)
	=&\tr_{\{1,2,3\}}\Big(e^{i\epsilon K_{D}}U_D  U_{G2} U_{G1} ( \ket{\psi_x^\text{in}}\bra{\psi_x^\text{in}} \otimes \ket{000}\bra{000}) U_{G1}^\dagger U_{G2}^\dagger U_D^\dagger e^{-i\epsilon K_{D}}\Big)\\
	=&\tr_{\{1,2,3\}}\Big(U_D  U_{G2} U_{G1} ( \ket{\psi_x^\text{in}}\bra{\psi_x^\text{in}} \otimes \ket{000}\bra{000}) U_{G1}^\dagger U_{G2}^\dagger  U_D^\dagger \\
	& +i\epsilon\ \left[K_{D}, U_D  U_{G2} U_{G1} ( \ket{\psi_x^\text{in}}\bra{\psi_x^\text{in}} \otimes \ket{000}\bra{000}) U_{G1}^\dagger U_{G2}^\dagger  U_D^\dagger\right] +\mathcal{O}(\epsilon^2)\Big)\\	
	=&\rho_{\mathrm{out}}^{G+D}(t)\\
	&+i\epsilon\ \tr_{\{1,2,3\}}\Big(\left[K_{D}, U_D U_{G2} U_{G1} ( \ket{\psi_x^\text{in}}\bra{\psi_x^\text{in}} \otimes \ket{000}\bra{000}) U_{G1}^\dagger U_{G2}^\dagger  U_D^\dagger\right]  \Big)\\
	& +\mathcal{O}(\epsilon^2).
\end{align*}
The update of the generator, assuming the discriminator is fixed, can be written as
\begin{align*}
	\rho_{\mathrm{out 2}}^{G+D}(t+\epsilon)
	=&\tr_{\{1,2,3\}}\Big(U_D e^{i\epsilon K_{G2}}U_{G2} e^{i\epsilon K_{G1}}U_{G1} ( \ket{\psi_x^\text{in}}\bra{\psi_x^\text{in}} \otimes \ket{000}\bra{000})\\
	& U_{G1}^\dagger e^{-i\epsilon K_{G1}} U_{G2}^\dagger e^{-i\epsilon K_{G2}}  U_D^\dagger \Big)\\
	=&\tr_{\{1,2,3\}}\Big(U_D \Big(   U_{G2} U_{G1} ( \ket{\psi_x^\text{in}}\bra{\psi_x^\text{in}} \otimes \ket{000}\bra{000}) U_{G1}^\dagger U_{G2}^\dagger	\\
	&+ i\epsilon \ U_{G2} \left[K_{G1}, U_{G1} ( \ket{\psi_x^\text{in}}\bra{\psi_x^\text{in}} \otimes \ket{000}\bra{000}) U_{G1}^\dagger  \right] U_{G2}^\dagger \\
	&+ i\epsilon \left[K_{G2}, U_{G2} U_{G1} ( \ket{\psi_x^\text{in}}\bra{\psi_x^\text{in}} \otimes \ket{000}\bra{000}) U_{G1}^\dagger U_{G2}^\dagger  \right]
	\Big) U_D^\dagger +\mathcal{O}(\epsilon^2) \Big) \\		
	=&\rho_{\mathrm{out 2}}^{G+D}(t)\\
	&+i\epsilon\ \tr_{\{1,2,3\}}\Big(U_D \Big(  
	U_{G2} \left[K_{G1}, U_{G1} ( \ket{\psi_x^\text{in}}\bra{\psi_x^\text{in}} \otimes \ket{000}\bra{000}) U_{G1}^\dagger  \right] U_{G2}^\dagger \\
	&+ \left[K_{G2}, U_{G2} U_{G1} ( \ket{\psi_x^\text{in}}\bra{\psi_x^\text{in}} \otimes \ket{000}\bra{000}) U_{G1}^\dagger U_{G2}^\dagger  \right]
	\Big) U_D^\dagger  \Big)+\mathcal{O}(\epsilon^2).
\end{align*}

To derive the update matrices in general we assume in the following a generator consisting of unitaries $U_1^1 \dots U_{m_g}^g$ and a discriminator built of unitaries $U_1^{g+1}\dots U_{m_{L+1}}^{L+1}$. The update matrices $K_j^l$ update the generator, if $l\le g$ for the number of perceptron layers $g$ of the generator. Otherwise, the matrices $K_j^l$ describe discriminator updates.

\begin{prop}
	\label{prop:QGAN_K}
	The update matrix for a QGAN trained with pure states $\ket{\phi^\text{T}_x}$  has to be of the form
	\begin{equation*}
		K^l_j(t) = \frac{\eta 2^{m_{l-1}}i}{S}\sum_x\tr_\text{rest}\big(M^l_{j}(x,t)\big),
	\end{equation*}
	where 
	\begin{align*}
		M_j^l =& \Big[ U_{j}^{l} \dots U_{1}^{1} \ ( \ket{\psi_x^\text{in}}\bra{\psi_x^\text{in}} \otimes \ket{0...0}\bra{0...0})  U_{1}^{1 \dagger} \dots U_{j}^{l \dagger}, \\
		&U_{j+1}^{l\dagger}\dots U_{m_{L+1}}^{L+1 \dagger} \left(\mathbbm{1}_\mathrm{in+hid}\otimes \ket{1}\bra{1}\right)U_{m_{L+1}}^{L+1 } \dots U_{l+1}^{l}\Big]
	\end{align*}
	for $l\le g$ and 
	\begin{align*}
		M_j^l =& \Big[  U_{j}^{l} \dots U_{1}^{g+1} \left(\ket{\phi^T}\bra{\phi^T} \otimes \ket{0...0}\bra{0...0}\right) U_{1}^{g+1 \dagger}\dots U_{j}^{l \dagger}   \\
		&- U_{j}^{l} \dots U_{1}^{g+1} U_{m_g}^{g} \dots U_{1}^{1} \ ( \ket{\psi_x^\text{in}}\bra{\psi_x^\text{in}} \otimes \ket{0...0}\bra{0...0})   U_{1}^{1 \dagger} \dots U_{m_{g}}^{g\dagger}  U_{1}^{g+1 \dagger}\dots U_{j}^{l \dagger} ,\\
		&U_{j+1}^{l\dagger}\dots U_{m_{L+1}}^{L+1 \dagger} \left(\mathbbm{1}_\mathrm{in+hid}\otimes \ket{1}\bra{1}\right)U_{m_{L+1}}^{L+1 } \dots U_{l+1}^{l}\Big]
	\end{align*}
	Here, $U_j^l$ is assigned to the $j$th perceptron acting on layers $l-1$and $l$, $g$ is the number of perceptron layers of the generator, and $\eta$ is the learning rate.
\end{prop}
\begin{proof}
	First, we compute the output state of the discriminator after an update with $K_D$. Note that in the following the unitaries act on the current layers, e.g. $U_1^l$ denotes actually $U_1^l\otimes \mathbbm{1}^l_{2,3,\dots,m_l}$. We fix the generator. To derive the update for the discriminator, we need the state when it is fed with the training data, i.e.\
	\begin{align*}
		\rho_{\mathrm{out}}^{D}(t+\epsilon)
		=&\tr_\mathrm{in(D)+hid}\Big(e^{i\epsilon K_{m_{L+1}}^{L+1}}U_{m_{L+1}}^{L+1} \dots e^{i\epsilon K_{1}^{g+1}}U_{1}^{g+1} \ \left(\ket{\phi^T}\bra{\phi^T} \otimes \ket{0...0}\bra{0...0}\right)  \\
		&  U_{1}^{g+1 \dagger}e^{-i\epsilon K_{1}^{g+1}} \dots U_{m_{L+1}}^{L+1 \dagger}e^{-i\epsilon K_{m_{L+1}}^{L+1}}\Big)\\
		=&\rho_{\mathrm{out}}^{D}(t)+i\epsilon\ \tr_\mathrm{in(D)+hid}\Big(
		\big[K_{m_{L+1}}^{L+1},U_{m_{L+1}}^{L+1} \dots U_{1}^{g+1} \left(\ket{\phi^T}\bra{\phi^T} \otimes \ket{0...0}\bra{0...0}\right)\\
		& U_{1}^{g+1 \dagger} \dots U_{m_{L+1}}^{L+1 \dagger} \big]+\dots + U_{m_{L+1}}^{L+1} \dots U_{2}^{g+1} \big[K_{1}^{g+1}, U_{1}^{g+1} \ \ket{\phi^T}\bra{\phi^T} \\
		&\otimes \ket{0...0}\bra{0...0} \ U_{1}^{g+1 \dagger} \big]  U_{2}^{g+1 \dagger} \dots U_{m_{L+1}}^{L+1 \dagger}\Big)+\mathcal{O}(\epsilon^2),
	\end{align*}
	and for the case it gets an input state from the generator, that is
	\begin{align*}
		\rho_{\mathrm{out}}^{G+D}(t+\epsilon)
		=&\tr_\mathrm{in(G)+hid}\Big(e^{i\epsilon K_{m_{L+1}}^{L+1}}U_{m_{L+1}}^{L+1} \dots e^{i\epsilon K_{1}^{g+1}}U_{1}^{g+1 } U_{m_g}^{g}\dots U_1^1  ( \ket{\psi_x^\text{in}}\bra{\psi_x^\text{in}} \\
		&\otimes \ket{0...0}\bra{0...0})   U_1^{1\dagger} \dots U_{m_g}^{g\dagger}   U_{1}^{g+1\dagger}e^{-i\epsilon K_{1}^{g+1}} \dots U_{m_{L+1}}^{L+1\dagger}e^{-i\epsilon K_{m_{L+1}}^{L+1}}\Big)\\
		=&\rho_{\mathrm{out}}^{D}(t)+i\epsilon\ \tr_\mathrm{in(G)+hid}\Big(
		\big[K_{m_{L+1}}^{L+1},U_{m_{L+1}}^{L+1} \dots U_{1}^{1} \ ( \ket{\psi_x^\text{in}}\bra{\psi_x^\text{in}}\\
		& \otimes \ket{0...0}\bra{0...0})  U_{1}^{1 \dagger} \dots U_{m_{L+1}}^{L+1 \dagger} \big]+\dots \\
		&+ U_{m_{L+1}}^{L+1} \dots U_{2}^{g+1} \big[K_{1}^{g+1}, U_{1}^{g+1} \ U_{m_g}^{g}\dots U_1^1  ( \ket{\psi_x^\text{in}}\bra{\psi_x^\text{in}} \otimes \ket{0...0}\bra{0...0}) \\
		&  U_1^{1\dagger} \dots U_{m_g}^{g\dagger} \ U_{1}^{g+1 \dagger} \big] U_{2}^{g+1 \dagger} \dots U_{m_{L+1}}^{L+1 \dagger} \Big) \\
		& +\mathcal{O}(\epsilon^2).	
	\end{align*}
	The derivative of the discriminator loss function has the following form:
	\begin{align*}
		\frac{d\mathcal{L}_D}{dt}=&\lim_{\epsilon\rightarrow 0}\frac{\mathcal{L}_D(t)+i\epsilon\frac{1}{S} \sum_{x=1}^S\bra{1}\tr_\mathrm{in+hid}(\dots)\ket{1}-\mathcal{L}_D(t)}{\epsilon}\\
		=&\frac{i}{S}\ \sum_{x=1}^S\tr_\mathrm{in+hid}\Big(\mathbbm{1}_\mathrm{in+hid}\otimes \ket{1}\bra{1}\Big(\Big(
		\big[K_{m_{L+1}}^{L+1},U_{m_{L+1}}^{L+1} \dots U_{1}^{g+1} \ket{\phi^T}\bra{\phi^T} \\
		&\otimes \ket{0...0}\bra{0...0} U_{1}^{g+1 \dagger} \dots U_{m_{L+1}}^{L+1 \dagger} \big]  +\hdots \\
		& + U_{m_{L+1}}^{L+1} \dots U_{2}^{g+1} \left[K_{1}^{g+1}, U_{1}^{g+1} \ \left(\ket{\phi^T}\bra{\phi^T} \otimes \ket{0...0}\bra{0...0}\right) \ U_{1}^{g+1 \dagger} \right]\\
		&  U_{2}^{g+1 \dagger} \dots U_{m_{L+1}}^{L+1 \dagger}\Big)  \\
		&- \Big(
		\big[K_{m_{L+1}}^{L+1},U_{m_{L+1}}^{L+1} \dots U_{1}^{1} \ ( \ket{\psi_x^\text{in}}\bra{\psi_x^\text{in}} \\
		& \otimes \ket{0...0}\bra{0...0}) U_{1}^{1 \dagger} \dots U_{m_{L+1}}^{L+1 \dagger} \big]+\dots \\
		& + U_{m_{L+1}}^{L+1} \dots U_{2}^{g+1} \big[K_{1}^{g+1}, U_{1}^{g+1} \ U_{m_g}^{g}\dots U_1^1  ( \ket{\psi_x^\text{in}}\bra{\psi_x^\text{in}} \otimes \ket{0...0}\bra{0...0}) \\
		&  U_1^{1\dagger} \dots U_{m_g}^{g\dagger} \ U_{1}^{g+1 \dagger} \big] U_{2}^{g+1 \dagger} \dots U_{m_{L+1}}^{L+1 \dagger} \Big)\Big)\Big) \\
		=&\frac{i}{S}\ \sum_{x=1}^S\tr_\mathrm{in+hid}\Big(\\
		&\Big[  U_{m_{L+1}}^{L+1} \dots U_{1}^{g+1} \left(\ket{\phi^T}\bra{\phi^T} \otimes \ket{0...0}\bra{0...0}\right) U_{1}^{g+1 \dagger} \dots U_{m_{L+1}}^{L+1 \dagger}\\
		&-U_{m_{L+1}}^{L+1} \dots U_{1}^{1} \ ( \ket{\psi_x^\text{in}}\bra{\psi_x^\text{in}} \otimes \ket{0...0}\bra{0...0})  U_{1}^{1 \dagger} \dots U_{m_{L+1}}^{L+1 \dagger} ,\\
		&\mathbbm{1}_\mathrm{in+hid}\otimes \ket{1}\bra{1}\Big]K_{m_{L+1}}^{L+1}+\dots +\Big[ U_{1}^{g+1} \left(\ket{\phi^T}\bra{\phi^T} \otimes \ket{0...0}\bra{0...0}\right) U_{1}^{g+1 \dagger}\\
		& - U_{1}^{g+1} U_{m_g}^{g} \dots U_{1}^{1} \ ( \ket{\psi_x^\text{in}}\bra{\psi_x^\text{in}} \otimes \ket{0...0}\bra{0...0}) \\
		&   U_{1}^{1 \dagger} \dots U_{m_{g}}^{g\dagger}  U_{1}^{g+1\dagger} , U_{m_{L+1}}^{L+1\dagger}\dots U_{2}^{g+1\dagger} \mathbbm{1}_\mathrm{in+hid}\otimes \ket{1}\bra{1}U_{2}^{g+1} \dots U_{m_{L+1}}^{L+1}\Big]K_{1}^{g+1}\Big)\\
		=&\frac{i}{S}\ \sum_{x=1}^S\tr_\mathrm{in+hid}\left(M_{m_{L+1}}^{L+1}K_{m_{L+1}}^{L+1}+\dots+M_{1}^{g+1}K_{1}^{g+1}\right).
	\end{align*}
	Note that at this point $ \ket{\phi^T}\bra{\phi^T} \otimes \ket{0...0}\bra{0...0}$ denotes $\mathbbm{1}_{in(G)+hid(G)} \otimes  \otimes   \ket{\phi^T}\bra{\phi^T} \otimes \ket{0...0}\bra{0...0} \otimes \mathbbm{1}_G \otimes \dots \otimes \mathbbm{1}_G $, to match the dimension of the other summand.
	
	Until here we fixed the generator. Now we study the second part of the algorithm: the generator is fixed instead. Using the state
	\begin{align*}
		\rho_{\mathrm{out 2}}^{G+D}(s+\epsilon)
		=&\tr_\mathrm{in(G)+hid}\Big(U_{m_{L+1}}^{L+1 } \dots U_{1}^{g+1}  \ e^{i\epsilon K_{m_{g}}^{g}}U_{m_{g}}^{g} \dots e^{i\epsilon K_{1}^{1}}U_{1}^{1 } \ ( \ket{\psi_x^\text{in}}\bra{\psi_x^\text{in}} \\
		&\otimes \ket{0...0}\bra{0...0})   U_{1}^{1\dagger} e^{-i\epsilon K_{1}^{1}}\dots U_{m_{g}}^{g\dagger} e^{-i\epsilon K_{m_{g}}^{g}}U_{1}^{g+1\dagger}\dots U_{m_{L+1}}^{L+1 \dagger}\Big)\\
		=&\rho_{\mathrm{out}}^{D}(t)+i\epsilon\ \tr_\mathrm{in()+hid}\Big(
		U_{m_{L+1}}^{L+1 } \dots U_{1}^{g+1}\big[K_{m_{g}}^{g},U_{m_{g}}^{g} \dots U_{1}^{1} \ ( \ket{\psi_x^\text{in}}\bra{\psi_x^\text{in}} \\
		& \otimes \ket{0...0}\bra{0...0})  U_{1}^{1 \dagger} \dots U_{m_{g}}^{g \dagger} \big] U_{1}^{g+1\dagger}\dots U_{m_{L+1}}^{L+1 \dagger} +\dots \\
		&+ U_{m_{L+1}}^{L+1} \dots U_{2}^{1} \left[K_{1}^{1}, U_{1}^{1} \ ( \ket{\psi_x^\text{in}}\bra{\psi_x^\text{in}} \otimes \ket{0...0}\bra{0...0}) \ U_{1}^{1 \dagger} \right] U_{2}^{1 \dagger} \dots U_{m_{L+1}}^{L+1 \dagger} \Big) \\
		& +\mathcal{O}(\epsilon^2)	
		G\end{align*}
	the derivative of the loss function for training the generator becomes
	\begin{align*}
		\frac{d\mathcal{L}_G}{dt}=&\lim_{\epsilon\rightarrow 0}\frac{\mathcal{L}_G(t)+i\epsilon\frac{1}{S} \sum_x\bra{1}\tr_\mathrm{in+hid}(\dots)\ket{1}-\mathcal{L}_G(t)}{\epsilon}\\
		=&\frac{i}{S}\ \sum_{x=1}^S\tr\Big(\mathbbm{1}_\mathrm{in+hid}\otimes \ket{1}\bra{1}\Big(\Big(
		U_{m_{l+1}}^{l+1 } \dots U_{1}^{g+1}\big[K_{m_{g}}^{g},U_{m_{g}}^{g} \dots U_{1}^{1} \ ( \ket{\psi_x^\text{in}}\bra{\psi_x^\text{in}} \\
		&\otimes \ket{0...0}\bra{0...0})  U_{1}^{1 \dagger} \dots U_{m_{g}}^{g \dagger} \big] U_{1}^{g+1\dagger}\dots U_{m_{l+1}}^{l+1 \dagger} +\dots \\
		&+ U_{m_{l+1}}^{l+1} \dots U_{2}^{1} \left[K_{1}^{1}, U_{1}^{1} \ ( \ket{\psi_x^\text{in}}\bra{\psi_x^\text{in}} \otimes \ket{0...0}\bra{0...0}) \ U_{1}^{1 \dagger} \right] U_{2}^{1 \dagger} \dots U_{m_{l+1}}^{l+1 \dagger} \Big)   \Big)\Big) \\
		=&\frac{i}{S}\ \sum_{x=1}^S\tr\Big(\\
		&\Big[ U_{m_{g}}^{g} \dots U_{1}^{1} \ ( \ket{\psi_x^\text{in}}\bra{\psi_x^\text{in}} \otimes \ket{0...0}\bra{0...0})  U_{1}^{1 \dagger} \dots U_{m_{g}}^{g \dagger}, \\
		&U_{1}^{g+1\dagger}\dots U_{m_{l+1}}^{l+1 \dagger} \left(\mathbbm{1}_\mathrm{in+hid}\otimes \ket{1}\bra{1}\right) U_{m_{l+1}}^{l+1 } \dots U_{1}^{g+1}\Big]K_{m_{g}}^{g}+\dots \\
		&+\Big[ U_{1}^{1} \ ( \ket{\psi_x^\text{in}}\bra{\psi_x^\text{in}} \otimes \ket{0...0}\bra{0...0})  U_{1}^{1 \dagger} , \\
		&U_{2}^{1 \dagger} \dots U_{m_{l+1}}^{l+1 \dagger}  \left(\mathbbm{1}_\mathrm{in+hid}\otimes \ket{1}\bra{1}\right) U_{m_{l+1}}^{l+1} \dots U_{2}^{1}\Big]K_{1}^{1}\Big)\\
		\equiv&\frac{i}{S}\ \sum_{x=1}^S\tr\left(M_{m_{g}}^{g}K_{m_{g}}^{g}+\dots+M_{1}^{1}K_{1}^{1}\right).
	\end{align*}
	
	In both updates, we parametrise the parameter matrices analogously as
	\begin{equation*}
		K_j^l(t)=\sum_{\alpha_1,\alpha_2,\dots,\alpha_{m_{l-1}},\beta}K^l_{j,\alpha_1,\dots,\alpha_{m_{l-1}},\beta}(t)\left(\sigma^{\alpha_1}\otimes\ \dots\ \otimes\sigma^{\alpha_{m_{l-1}}}\otimes\sigma^\beta\right),
	\end{equation*}
	where the $\alpha_i$ denote the qubits in the previous layer and $\beta$ denotes the current qubit in layer $l$. To achieve the maximum of the loss function as a function of the parameters \emph{fastest}, we maximise $\frac{d\mathcal{L}}{dt}$. Since this is a linear function, the extrema are at $\pm\infty$. To ensure that we get a finite solution, we introduce a Lagrange multiplier $\lambda\in\mathbbm{R}$. Hence, to find $K_j^l$ we have to solve the following maximisation problem (here for the discriminator update, the update for the generator is analogous):
	\begin{align*}
		\max_{K^l_{j,\alpha_1,\dots,\beta}}&\left(\frac{dC(t)}{dt}-\lambda\sum_{\alpha_i,\beta}{K^l_{j,\alpha_1,\dots,\beta}}(t)^2\right)\\
		&=\max_{K^l_{j,\alpha_1,\dots,\beta}}\left(\frac{i}{S}\sum_{x=1}^S \tr\left(M_{m_{L+1}}^{L+1}K_{m_{L+1}}^{L+1}+\dots+M_{1}^{g+1}K_{1}^{g+1}\right)-\lambda\sum_{\alpha_1,\dots,\beta}{K^l_{j,\alpha_1,\dots,\beta}}(t)^2\right)\\
		&=\max_{K^l_{j,\alpha_1,\dots,\beta}}\Big(\frac{i}{S}\sum_{x=1}^S\tr_{\alpha_1,\dots,\beta}\left(\tr_\mathrm{rest}\left(M_{m_{L+1}}^{L+1}K_{m_{L+1}}^{L+1}+\dots+M_{1}^{g+1}K_{1}^{g+1}\right)\right)\\
		&-\lambda\sum_{\alpha_1,\dots,\beta}{K^l_{j,\alpha_1,\dots,\beta}}(t)^2\Big).
	\end{align*}
	
	Taking the derivative with respect to $K^l_{j,\alpha_1,\dots,\beta}$ yields
	\begin{align*}
		\frac{i}{S}\sum_{x=1}^S\tr_{\alpha_1,\dots,\beta}\left(\tr_\mathrm{rest}\left(M_j^l(t)\right)\left(\sigma^{\alpha_1}\otimes\ \dots\ \otimes\sigma^\beta\right)\right)-2\lambda K^l_{j,\alpha_1,\dots,\beta}(t)=0,
	\end{align*}
	hence,
	\begin{align*}
		K^l_{j,\alpha_1,\dots,\beta}(t)=\frac{i}{2S\lambda}\sum_{x=1}^S\tr_{\alpha_1,\dots,\beta}\left(\tr_\mathrm{rest}\left(M_j^l(t)\right)\left(\sigma^{\alpha_1}\otimes\ \dots\ \otimes\sigma^\beta\right)\right)
	\end{align*}
	This yields the matrix 
	\begin{align*}
		K_j^l(t)&=\sum_{\alpha_1,\dots,\beta}K^l_{j,\alpha_1,\dots,\beta}(t)\left(\sigma^{\alpha_1}\otimes\ \dots\ \otimes\sigma^\beta\right)\\
		&=\frac{i}{2S\lambda}\sum_{\alpha_1,\dots,\beta}\sum_{x=1}^S\tr_{\alpha_1,\dots,\beta}\left(\tr_\mathrm{rest}\left(M_j^l(t)\right)\left(\sigma^{\alpha_1}\otimes\ \dots\ \otimes\sigma^\beta\right)\right)\left(\sigma^{\alpha_1}\otimes\ \dots\ \otimes\sigma^\beta\right)\\
		&=\frac{\eta2^{m_{l-1}}i}{2S}\sum_{x=1}^S\tr_\mathrm{rest}\left(M_j^l(t)\right),
	\end{align*}
	where $\eta=1/\lambda$ is the learning rate and $\tr_\text{rest}$ traces out all qubits that the perceptron unitary $U_j^l$ does not act on.
	
	Notice again that $K_j^l$ updates the generator, if $j\le g$ for the number of layers $g$ of the generator. The definition of $M_j^l$ is
	\begin{align*}
		M_j^l =& \Big[ U_{j}^{l} \dots U_{1}^{1} \ ( \ket{\psi_x^\text{in}}\bra{\psi_x^\text{in}} \otimes \ket{0...0}\bra{0...0})  U_{1}^{1 \dagger} \dots U_{j}^{l \dagger}, \\
		&U_{j+1}^{l\dagger}\dots U_{m_{L+1}}^{L+1 \dagger} \left(\mathbbm{1}_\mathrm{in+hid}\otimes \ket{1}\bra{1}\right)U_{m_{L+1}}^{L+1 } \dots U_{l+1}^{l}\Big]
	\end{align*}
	for $l\le g$ and 
	\begin{align*}
		M_j^l =& \Big[  U_{j}^{l} \dots U_{1}^{g+1} \left(\ket{\phi^T}\bra{\phi^T} \otimes \ket{0...0}\bra{0...0}\right) U_{1}^{g+1 \dagger}\dots U_{j}^{l \dagger}   \\
		&- U_{j}^{l} \dots U_{1}^{g+1} U_{m_g}^{g} \dots U_{1}^{1} \ ( \ket{\psi_x^\text{in}}\bra{\psi_x^\text{in}} \otimes \ket{0...0}\bra{0...0})   U_{1}^{1 \dagger} \dots U_{m_{g}}^{g\dagger}  U_{1}^{g+1 \dagger}\dots U_{j}^{l \dagger} ,\\
		&U_{j+1}^{l\dagger}\dots U_{m_{L+1}}^{L+1 \dagger} \left(\mathbbm{1}_\mathrm{in+hid}\otimes \ket{1}\bra{1}\right)U_{m_{L+1}}^{L+1 } \dots U_{l+1}^{l}\Big]
	\end{align*}
	else.
\end{proof}

\section{Implementation of the DQNN\textsubscript{Q} as a PQC\label{section:dqnn_q_implementation_details}}
The DQNN\textsubscript{Q} intends to realise each neuron as a separate qubit. Thus, implementing the DQNN\textsubscript{Q} as a quantum circuit requires $M=\sum_{l=0}^{L+1} m_l$ qubits. This results in a $2^M$-dimensional Hilbert space $\mathcal{H}^{\otimes M}$ which is the tensor product of $M$ single qubit Hilbert spaces.

The main task of implementing the DQNN described by \cref{eq:DQNN_rhoOut} is to find an appropriate realisation of the quantum perceptron $U^l_j$ which is a general unitary acting on $m_{l-1}+1$ qubits. For the simulation on a classical computer it is sufficient to abstractly define the unitary matrix and update its entries during the training. However, to execute the DQNN on a quantum computer, a concrete realisation in the form of parameterised quantum gates is necessary to build. Once the parameterised quantum gates for representing the quantum perceptron are chosen, the full PQC can be built by composing the respective quantum perceptrons from all layers. When thinking about possible candidates for parameterised quantum gates, two objectives have to be well-balanced: on the one hand, the final realisation of the quantum perceptron should be as universal as possible, while on the other hand, the number of quantum gates and parameters should be kept as small as possible. If either one of these objectives is neglected, the DQNN\textsubscript{Q} will not perform as well as its classically simulated model.

Any arbitrary two-qubit unitary can be expressed by a two-qubit canonical gate and twelve single-qubit gates \cite{Crooks2019}. The two-qubit canonical gate is defined via three parameters as:
\begin{align}
\begin{split}
\text{CAN}(t_x,t_y,t_z) &= e^{-i\frac{\pi}{2}t_x X \otimes X}e^{-i\frac{\pi}{2}t_y Y \otimes Y}e^{-i\frac{\pi}{2}t_z Z \otimes Z} \\
&= \text{RXX}(t_x\pi)\;\text{RYY}(t_y\pi)\;\text{RZZ}(t_z\pi)
\end{split}
\end{align}
where $X = \begin{psmallmatrix} 0&1\\ 1&0 \end{psmallmatrix}$, $Y = \begin{psmallmatrix} 0&-i\\ i&0 \end{psmallmatrix}$, $Z = \begin{psmallmatrix} 1&0\\ 0&-1 \end{psmallmatrix}$ are the Pauli matrices, the RXX/RYY/RZZ gates are parameterised two qubit gates commonly available in quantum computing libraries, and $t_{x,y,z}\in\real$ are the parameters. The necessary single qubit gates are parameterised Pauli-$Y$ and Pauli-$Z$ operators. These are equivalent to the following rotations around the $y$- and the $z$-axis:
\begin{align}\label{eq:single_qubit_rotations}
\begin{split}
Y^t \simeq R_Y(\pi t) = e^{-i\frac{\pi}{2}tY} \\
Z^t \simeq R_Z(\pi t) = e^{-i\frac{\pi}{2}tZ}
\end{split}
\end{align}
up to a phase factor which is indicated by $\simeq$. By executing the two-qubit canonical gate in addition to prepending and appending three single qubit gates to each qubit in the following form:
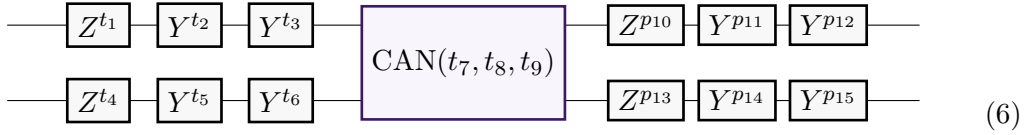
\begin{equation}\label{eq:universal_two_qubit_gate}
\begin{tikzpicture}[xscale=1.2]
\draw (0,1)-- (10,1);
\draw (0,0)-- (10,0);
\node[operator0,minimum height=0.5cm] at (1,1){$Z^{t_1}$};
\node[operator0,minimum height=0.5cm] at (2,1){$Y^{t_2}$};
\node[operator0,minimum height=0.5cm] at (3,1){$Y^{t_3}$};
\node[operator0,minimum height=0.5cm] at (1,0){$Z^{t_4}$};
\node[operator0,minimum height=0.5cm] at (2,0){$Y^{t_5}$};
\node[operator0,minimum height=0.5cm] at (3,0){$Y^{t_6}$};
\node[operator1,minimum height=1.5cm] at (5,0.5){$\can(t_7,t_8,t_9)$};
\node[operator0,minimum height=0.5cm] at (7,1){$Z^{p_{10}}$};
\node[operator0,minimum height=0.5cm] at (8,1){$Y^{p_{11}}$};
\node[operator0,minimum height=0.5cm] at (9,1){$Y^{p_{12}}$};
\node[operator0,minimum height=0.5cm] at (7,0){$Z^{p_{13}}$};
\node[operator0,minimum height=0.5cm] at (8,0){$Y^{p_{14}}$};
\node[operator0,minimum height=0.5cm] at (9,0){$Y^{p_{15}}$};
\end{tikzpicture}
\end{equation}
any arbitrary two-qubit gate can be performed. As a graphical simplification, the used sequence $Z$-$Y$-$Z$ of single-qubit gates can be expressed as the commonly used single qubit gate $u(t_1,t_2,t_3)$:
\begin{align}\label{eq:single_qubit_sequence}
u(t_1,t_2,t_3) = R_Z(t_2)R_Y(t_1)R_Z(t_3) =
\begin{pmatrix}
\cos (t_1/2) & -e^{it_3}\sin (t_1/2) \\
e^{it_2}\sin (t_1/2) & e^{i(t_2+t_3)}\cos (t_1/2)
\end{pmatrix}
\end{align}
where the different parameterisation compared to \Cref{eq:single_qubit_rotations} should be noted.

The quantum perceptron is not, in general, a two-qubit unitary. Therefore the universal two-qubit gate from \ref{eq:universal_two_qubit_gate} can not directly be used. When thinking about implementing the universal two-qubit gate, it is helpful to think about the task fulfilled by the quantum perceptron, which is to process the states of its input qubits and change the output qubit's state accordingly. This motivates the application of separate two-qubit gates on each input-output qubit pair. However, numerical studies have shown that it is sufficient and often advantageous to refrain from using the single-qubit sequence from \Cref{eq:single_qubit_sequence} and only use the two-qubit canonical gate as the direct realisation of the quantum perceptrons. In addition to realising the entire layer unitary $U^l$, i.e., all quantum perceptrons corresponding to layer $l$, the three-parameter single-qubit gate $u$ is prepended to all input qubits and appended to all output qubits. To append single-qubit gates on the input qubits is pointless, as these are no longer used. To prepend single-qubit gates on the input qubits has proven unnecessary in numerical studies.

The interpretation of the DQNN as a quantum circuit employing the previously discussed methods looks as follows. The first $m_0$ qubits are initialised in a given, possibly mixed state $\rho_\text{in}$, while all remaining qubits are initialised in the computational basis state $\ket{0}$. The quantum circuit and the general DQNN architecture are structured layer-wise and will therefore be described accordingly. The $u$ gates are applied first, layer by layer ($l=1,\dots,L+1$), to the respective $m_{l-1}$ input qubits. After that, the layer unitary $U^l = \prod _{j=m_l}^1U_j^l$ is applied to all input and output qubits. Here, $U^l_j$ is a sequence of $m_{l-1}$ CAN gates where the $i^\text{th}$ CAN gate acts on the $i^\text{th}$ input and the $j^\text{th}$ output qubit. After each layer $l$, the $m_{l-1}$ input qubits are neglected, i.e., they are just ignored for the rest of the quantum circuit. This layer's $m_l$ output qubits serve as the input qubits for the next layer $l+1$. By this, the partial trace of \Cref{eq:DQNN_rhoOut} is realised. After the output layer $L+1$, again, $u$ gates are applied to the remaining $m$ output qubits. Thus, the quantum circuit consists of $N_p = 3m + 3\sum ^{L+1}_{l=1} m_{l-1} (1+m_{l})$ parameters.

\begin{figure}
\centering
\begin{subfigure}[t]{0.35\linewidth}
\centering
\begin{tikzpicture}[scale=1]
\foreach \x in {-.5,.5} {
\draw[line0] (0,\x) -- (2,-1);
\draw[line1] (0,\x) -- (2,-1);
\draw[line0] (0,\x) -- (2,0);
\draw[line1] (0,\x) -- (2,0);
\draw[line0] (0,\x) -- (2,1);
}
\draw[line1,densely dotted] (0,-.5) -- (2,1);
\draw[line1, dash pattern=on 6pt off 2pt] (0,.5) -- (2,1);
\foreach \x in {-1,0,1} {
\draw[line0] (2,\x) -- (4,-0.5);
\draw[line2] (2,\x) -- (4,-0.5);
\draw[line0] (2,\x) -- (4,0.5);
\draw[line2] (2,\x) -- (4,0.5);
}
\node[perceptron0] at (0,-0.5) {};
\node[perceptron0] at (0,0.5) {};
\node[perceptron0] at (2,-1) {};
\node[perceptron0] at (2,0) {};
\node[perceptron0] at (2,1) {};
\node[perceptron0] at (4,-0.5) {};
\node[perceptron0] at (4,0.5) {};
\end{tikzpicture}
\end{subfigure}\begin{subfigure}[t]{0.59\linewidth}
\centering
\begin{tikzpicture}[]
\matrix[row sep=0.25cm, column sep=0.4cm] (circuit) {
\node(start3){$\ket{\phi^\text{in}}$};  
& \node[halfcross,label={\small 2}] (c13){};
& \node[operator0] (c23){$u^{\otimes 2}$};
& \node[]{}; 
& \node[dcross](end3){}; 
& \node[]{}; 
& \node[]{}; 
& \node[]{}; \\
\node(start2){$\ket{000}$};
& \node[halfcross,label={\small 3}] (c12){};
& \node[]{}; 
& \node[]{}; 
& \node[operator0] (c32){$u^{\otimes 3}$};
& \node[]{}; 
& \node[dcross](end2){}; 
& \node[]{};  \\
\node(start1){$\ket{00}$};
& \node[halfcross,label={\small 2}] (c11){};
& \node[]{}; 
& \node[]{}; 
& \node[]{}; 
& \node[]{}; 
& \node[operator0] (c41){$u^{\otimes 3}$};
& \node (end1){$\rho ^\text{out}$}; \\
};
\begin{pgfonlayer}{background}
\draw[] (start1) -- (end1)  
(start2) -- (end2)
(start3) -- (end3);
\node[operator1, minimum height=2cm] at (-.48,0.5) {$U^1$};
\node[operator2,minimum height=2cm] at (1.49,-.7) {$U^2$};
\end{pgfonlayer}
\end{tikzpicture}
\end{subfigure}
\vspace{0.5 cm}
\begin{subfigure}[t]{0.35\linewidth}
\centering
\begin{tikzpicture}[scale=0.8]
\draw[white] (0,2.2)-- (1,2.2);
\draw (0.2,1)-- (1.8,1);
\draw (0.2,0)-- (1.8,0);
\node[operator1,minimum height=1.5cm] at (1,0.5){$U^1$};
\node[] at (2.1,0.5){$=$};
\begin{scope}[xshift=2.5cm]
\node[rounded corners=.35cm,draw=color3,line width=1pt,dashed,minimum height=2.0cm,minimum width=2.1cm, fill=color3L] at (2.6,0.5){};
\node[color3] at (2.6,2.3){2-3-2$^+$};
\draw (0,1)-- (4,1);
\draw (0,0)-- (4,0);
\node[operator1,minimum height=1.5cm] at (0.7,0.5){$U^1_a$};
\node[operator0]at (2,1) {$u^{\otimes 2}$};
\node[operator0]at (2,0) {$u^{\otimes 3}$};
\node[operator1,minimum height=1.5cm] at (3.3,0.5){$U^1_b$};
\end{scope}
\end{tikzpicture}
\end{subfigure}
\begin{subfigure}[t]{0.59\linewidth}
\centering
\begin{tikzpicture}[yscale=.4,xscale=0.8]
\draw (0,4)-- (2,4);
\draw (0,3)-- (2,3);
\draw (0,2)-- (2,2);
\draw (0,1)-- (2,1);
\draw (0,0)-- (2,0);
\node[operator1,minimum height=2.0cm] at (1,2){$U^1_{a,b}$};
\node[] at (2.5,2){$=$};
\begin{scope}[xshift=3cm]
\draw[white] (0,4.5)-- (7,4.5);
\draw (0,4)-- (7,4);
\draw (0,3)-- (7,3);
\draw (0,2)-- (7,2);
\draw (0,1)-- (7,1);
\draw (0,0)-- (7,0);
\node[operator1,minimum height=1.2cm,line width=1pt, dash pattern=on 6pt off 2pt] at (1,3){};
\node[operator1,minimum height=0.8cm, line width=1pt, densely dotted] at (2,2.5){};
\node[operator1,minimum height=1.6cm] at (3,2.5){};
\node[operator1,minimum height=1.2cm] at (4,2){};
\node[operator1,minimum height=2.0cm] at (5,2){};
\node[operator1,minimum height=1.6cm] at (6,1.5){};
\draw[line0] (0,3)-- (1.5,3);
\draw (0,3)-- (1.5,3);
\draw[line0] (2.5,3)-- (3.5,3);
\draw (2.5,3)-- (3.5,3);
\draw[line0] (4.5,3)-- (5.5,3);
\draw (4.5,3)-- (5.5,3);
\draw[line0] (2.5,2)-- (7,2);
\draw (2.5,2)-- (7,2);
\draw[line0] (4.5,1)-- (7,1);
\draw (4.5,1)-- (7,1);
\end{scope}
\end{tikzpicture}
\end{subfigure}
\caption{An exemplary DQNN\textsubscript{Q} implementation as a parameterised quantum circuit suitable for the execution on NISQ devices. The unitaries $U^l$ implement the layer-to-layer transition from the layer $l-1$ to $l$. In the standard 2-3-2 network, $U^l$ consists of $m_{l-1}\cdot m_l$ two-qubit CAN gates. In the computationally more powerful 2-3-2$^+$ network, $U^l$ features additional gates as shown in the pink dashed box.}
\label{fig:dqnn_circuit_implementation}
\end{figure}
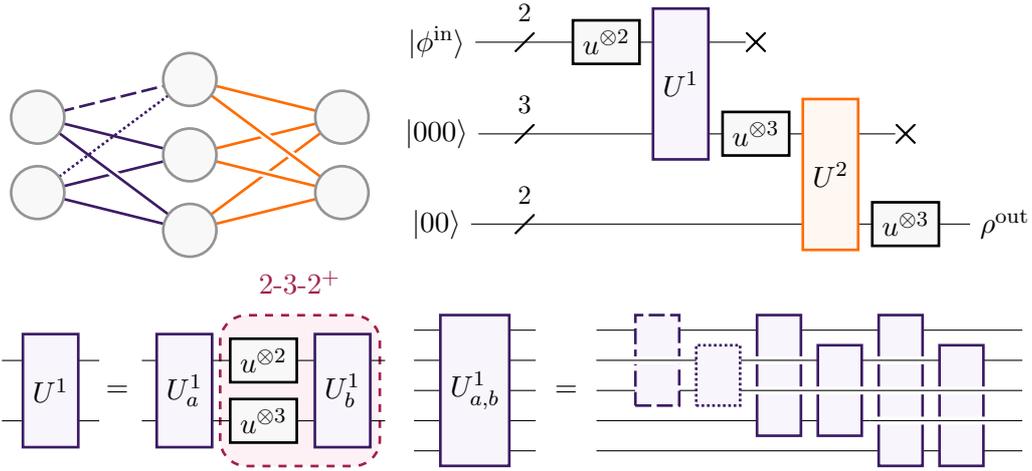

Due to the limitations of the current NISQ devices one is often interested in increasing the computational power of the DQNN\textsubscript{Q} without using additional qubits. In this case, the quantum perceptron can be modified such that the DQNN\textsubscript{Q}'s layer-to-layer transition gets computationally more powerful. This modification is denoted with a $^+$ as in 2-3-2$^+$. The corresponding DQNN\textsubscript{Q} is defined with the additional parameterised quantum gates shown in the pink dashed box in \cref{fig:dqnn_circuit_implementation}. The layer unitaries $U^1_a$ and $U^1_b$ share the same structure but are parameterised independently.

\section{Further numerical results}\label{apnx:numerics}

In \cref{section_results} we discussed the classical simulation of the DQGAN algorithm. In the following we extend the numerical examples of this section.

\begin{figure}[H]
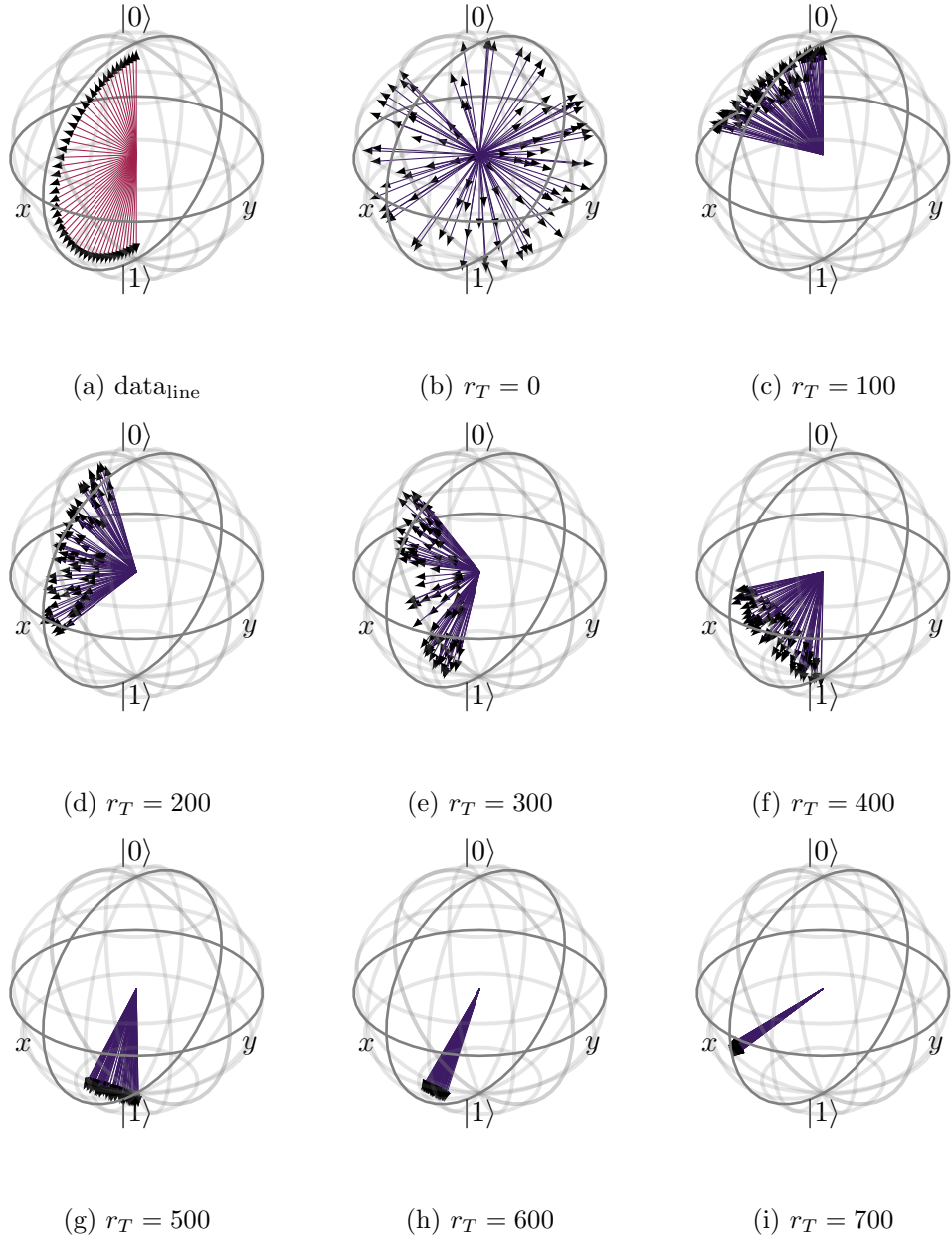

\centering
\begin{subfigure}{0.3\linewidth}
\input{numerics/T}
\subcaption{$\text{data}_\text{line}$}
\end{subfigure}
\begin{subfigure}{0.3\linewidth}
\input{numerics/0}
\subcaption{$r_T=0$}
\end{subfigure}
\begin{subfigure}{0.3\linewidth}
\input{numerics/100}
\subcaption{$r_T=100$}
\end{subfigure}

\begin{subfigure}{0.3\linewidth}
\input{numerics/200}
\subcaption{$r_T=200$}
\end{subfigure}
\begin{subfigure}{0.3\linewidth}
\input{numerics/300}
\subcaption{$r_T=300$}
\end{subfigure}
\begin{subfigure}{0.3\linewidth}
\input{numerics/400}
\subcaption{$r_T=400$}
\end{subfigure}

\begin{subfigure}{0.3\linewidth}
\input{numerics/500}
\subcaption{$r_T=500$}
\end{subfigure}
\begin{subfigure}{0.3\linewidth}
\input{numerics/600}
\subcaption{$r_T=600$}
\end{subfigure}
\begin{subfigure}{0.3\linewidth}
\input{numerics/700}
\subcaption{$r_T=700$}
\end{subfigure}
\caption{\textbf{Output of the generator.} To compare the output of the generator (b-i), during the training of a \protect\oneoneone DQGAN, to the data set $\text{data}_\text{line}$ (a) we plot the states in Bloch spheres.}
\label{fig:apdx_bloch}
\end{figure}

First of all, \cref{fig:apdx_bloch} gives an overview of the generator's different training situations following the training depicted in \cref{fig:GAN_line}. At every of these training steps we build a set of $100$ by the generator produced states and plot them in a Bloch sphere.

Secondly, in \label{fig:apdx_line} we train a 1-3-1 DQGAN with the training data
\begin{equation*}
\text{data}_\text{line'}=\left\{\frac{(N-x)\ket{000}+(x-1)\ket{001}}{||(N-x)\ket{000}+(x-1)\ket{001}||}\right\}_{x=1}^{N},
\end{equation*} for $N=50$. 

\begin{figure}[H]
\centering
\begin{tikzpicture}
\begin{axis}[
xmin=0,   xmax=20,
ymin=0,   ymax=2,
width=0.8\linewidth, 
height=0.5\linewidth,
grid=major,grid style={color0M},
xlabel= Training epochs $r_T$, 
xticklabels={-100,0,100,200,300,400,500,600,700,800,900,1000},
ylabel=$\mathcal{L}(t)$,legend pos=north east,legend cell align={left},legend style={draw=none,legend image code/.code={\filldraw[##1] (-.5ex,-.5ex) rectangle (0.5ex,0.5ex);}}]
\coordinate (0,0) ;
\addplot[mark size=1.5 pt,color=color2] table [x=step times epsilon, y=costFunctionDis, col sep=comma] {numerics/QGAN_50data10sv_100statData_100statData_1-3networkGen_3-1networkDis_lda1_ep0i01_rounds1000_roundsGen1_roundsDis1_connectedLine_training.csv};
\addlegendentry[mark size=10 pt,scale=1]{Training loss $\mathcal{L}_\text{D}$} 
\addplot[mark size=1.5 pt,color=color1] table [x=step times epsilon, y=costFunctionGen, col sep=comma] {numerics/QGAN_50data10sv_100statData_100statData_1-3networkGen_3-1networkDis_lda1_ep0i01_rounds1000_roundsGen1_roundsDis1_connectedLine_training.csv};
\addlegendentry[scale=1]{Training loss $\mathcal{L}_\text{G}$} 
\addplot[mark size=1.5 pt,color=color3] table [x=step times epsilon, y=costFunctionTest, col sep=comma] {numerics/QGAN_50data10sv_100statData_100statData_1-3networkGen_3-1networkDis_lda1_ep0i01_rounds1000_roundsGen1_roundsDis1_connectedLine_training.csv};
\addlegendentry[scale=1]{Validation loss $\mathcal{L}_\text{V}$} 
\end{axis}
\end{tikzpicture}
\caption{\textbf{Training a DQGAN.} The evolution of the training losses and validationloss during the training of a \protect\oneothreeone DQGAN in $r_T=1000$ epochs with $\eta=1$ and $\epsilon=0.01$ using $50$ data pairs of the data set $\text{data}_\text{line'}$ where $10$ are used as training states.}
\label{fig:apdx_line}
\end{figure}
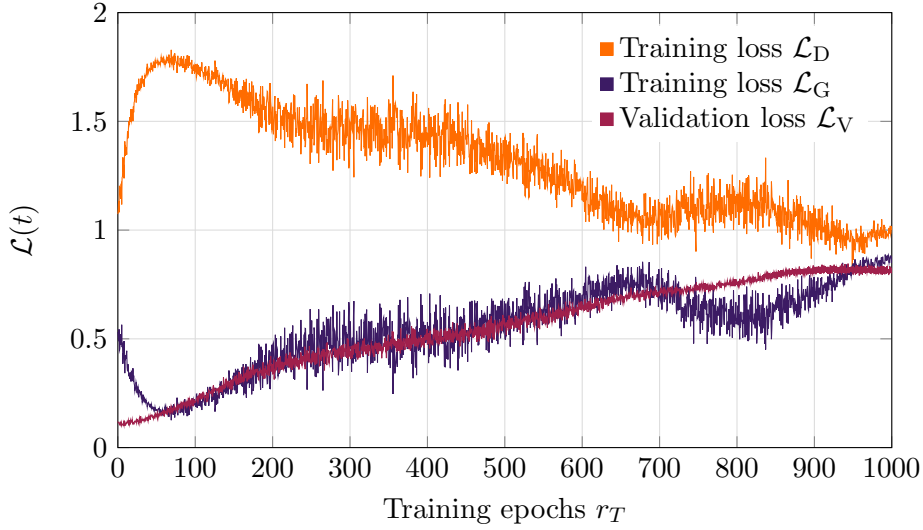

For a more comprehensive study, we averaged the histogram resulting after $200$ training rounds using ten independent training attempts and $10$ randomly chosen training states of $\text{data}_\text{line}$. \cref{fig:GAN_lineComp} shows that the diversity of the generator's output is good, since all elements in $\text{data}_\text{line}$ get produced quite equally.

Moreover, we build an equivalent plot with the difference of choosing randomly $10$ training states of $\text{data}_\text{cl}$, where
\begin{equation*}
\text{data}_\text{cl}= \left\{\frac{(2N-1)\ket{0}+(x-1)\ket{1}}{||(2N-1)\ket{0}+(x-1)\ket{1}||}\right\}_{x=1}^{\tfrac{N}{2}}\cup\left\{\frac{(2N-1)\ket{0}+(x-1)\ket{1}}{||(2N-1)\ket{0}+(x-1)\ket{1}||}\right\}_{x=\tfrac{3N}{2}}^{2N}. 
\end{equation*}

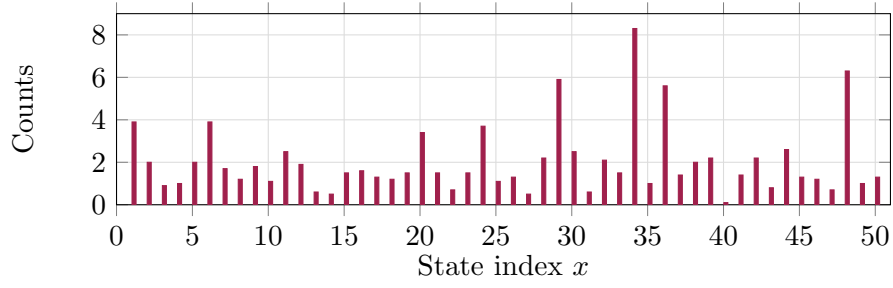
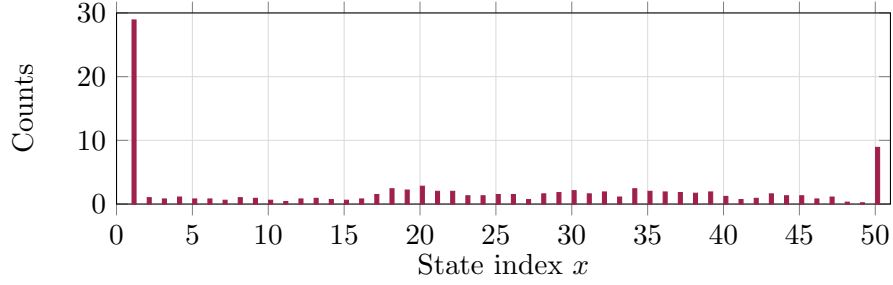
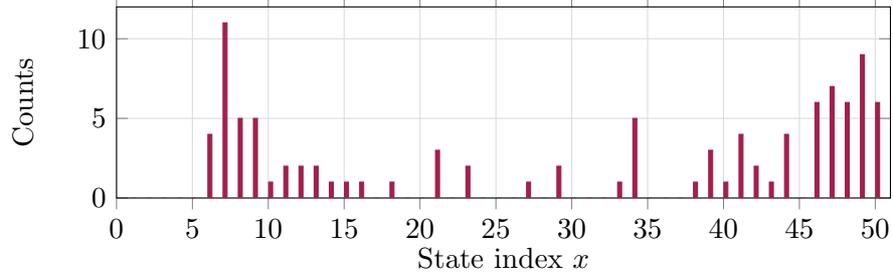
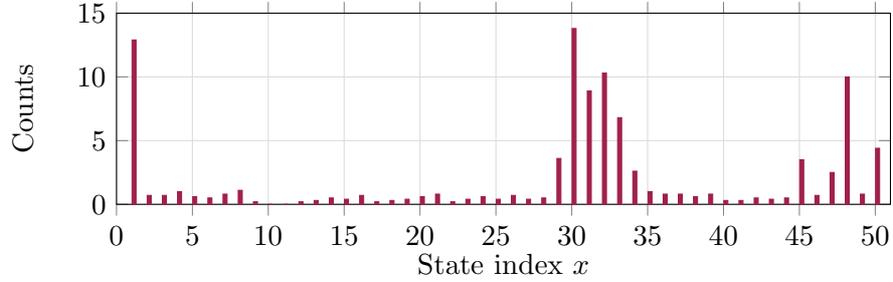
\begin{figure}[H]
\centering
\begin{subfigure}{\textwidth}\centering
\begin{tikzpicture}[scale=1]
\begin{axis}[
ybar,
bar width=1.5pt,
xmin=0,   xmax=51,
ymin=0,   ymax=9,
width=.8\linewidth, 
height=.28\linewidth,
grid=major,
grid style={color0M},
xlabel= State index $x$, 
ylabel=Counts,legend pos=north east,legend cell align={left}]
\addplot[color=color3, fill=color3] table [x=index,y=countTTMean, col sep=comma] {numerics/QGAN_50data10sv_100statData_100statData_1-1networkGen_1-1networkDis_lda1_ep0i01_rounds200_roundsGen1_roundsDis1_line_statMean.csv};
\end{axis}
\end{tikzpicture}
\caption{Line trained with DQNN.} \label{fig:GAN_lineComp}
\end{subfigure}
\begin{subfigure}{\textwidth}\centering
\begin{tikzpicture}[scale=1]
\begin{axis}[
ybar,
bar width=1.5pt,
xmin=0,   xmax=51,
ymin=0,   ymax=30,
width=.8\linewidth, 
height=.28\linewidth,
grid=major,
grid style={color0M},
xlabel= State index $x$, 
ylabel=Counts,legend pos=north east,legend cell align={left}]
\addplot[color=color3, fill=color3] table [x=index,y=countTTMean, col sep=comma] {numerics/QGAN_50data10sv_100statData_100statData_1-1networkGen_1-1networkDis_lda1_ep0i01_rounds200_roundsGen1_roundsDis1_CvsLi_statMean.csv};
\end{axis}
\end{tikzpicture}
\caption{Two clusters trained with DQNN.} \label{fig:GAN_ClusComp}
\end{subfigure}
\begin{subfigure}{\textwidth}\centering
\begin{tikzpicture}[scale=1]
\begin{axis}[
ybar,
bar width=1.5pt,
xmin=0,   xmax=51,
ymin=0,   ymax=12,
width=.8\linewidth, 
height=.28\linewidth,
grid=major,
grid style={color0M},
xlabel= State index $x$, 
ylabel=Counts,legend pos=north east,legend cell align={left}]
\addplot[color=color3, fill=color3] table [x=indexDataTest,y=countOutTest, col sep=comma] {numerics/dqnn_q_eq_cluster_epoch_200_vs.csv};
\end{axis}
\end{tikzpicture}
\caption{Two clusters trained with DQNN\textsubscript{Q}.} \label{fig:qgan_q_cluster}
\end{subfigure}
\begin{subfigure}{\textwidth}\centering
\begin{tikzpicture}[scale=1]
\begin{axis}[
ybar,
bar width=1.5pt,
xmin=0,   xmax=51,
ymin=0,   ymax=15,
width=.8\linewidth, 
height=.28\linewidth,
grid=major,
grid style={color0M},
xlabel= State index $x$, 
ylabel=Counts,legend pos=north east,legend cell align={left}]
\addplot[color=color3, fill=color3] table [x=index,y=countTTMean, col sep=comma] {numerics/QGAN_50data10sv_100statData_100statData_1-1networkGen_1-1networkDis_lda1_ep0i01_rounds200_roundsGen1_roundsDis1_conCvsLi_statMean.csv};
\end{axis}
\end{tikzpicture}
\caption{Two clusters plus $\frac{1}{\sqrt{2}}(\ket{0}+\ket{1})$ trained with DQNN.} \label{fig:GAN_Clus+Comp}
\end{subfigure}
\caption{\textbf{Diversity analysis of a DQGAN.} This plot describes the output's diversity of a \protect\oneoneone DQGAN (DQGAN\textsubscript{Q}) trained in 200 epochs with $\eta=1$ ($\eta_D=0.5,\eta_G=0.1$) and $\epsilon=0.01$ ($\epsilon = 0.25$) using $10$ quantum states of the data sets $\text{data}_\text{line}$ (a), $\text{data}_\text{cl}$ (b,c) and  $\text{data}_\text{cl+}$ (d) and compared the generator's output to the data set $\text{data}_\text{line}$.}
\label{fig:apnx_Div}
\end{figure}

\cref{fig:GAN_ClusComp} depicts the distribution of the generator's output after 200 training epochs of ten training attempts with $S=10$ randomly chosen training states. The generator does not produce all elements in $\text{data}_\text{line}$ equally often. Due to the average of ten independent training attempts, the states $\ket{0}$ and $\ket{1}$ are very prominent in this plot. Since the state $\ket{0}$ is produced more often, we assume that the training states randomly chosen in every training attempt the $S=10$ training states were more often states of the first part of the cluster.

Further, by removing one state of the data set $\text{data}_\text{cl}$ and replacing it by $\frac{1}{\sqrt{2}}(\ket{0}+\ket{1})$ we obtain the data set $\text{data}_\text{cl+}$. \cref{fig:GAN_Clus+Comp} shows the diversity of a generator resulting by training a DQGAN with $\text{data}_\text{cl+}$. We can see, that some states in the middle of the $x$-range are generated more often compared to the plot in \cref{fig:GAN_ClusComp}. However, the generator does not produce the state $\frac{1}{\sqrt{2}}(\ket{0}+\ket{1})$ ($x=25$) very often and the resulting peak in the histogram is rather shifted more in the direction of the  $\ket{1}$ state ($x=50$).

Additionally, we trained a DQGAN\textsubscript{Q} using the clustered data set $\text{data}_\text{cl}$ and tested the generator's diversity after $r_T=200$ training epochs for a single execution on the $\text{data}_\text{line}$. The results are depicted in \cref{fig:qgan_q_cluster} which show the generator's ability to extend the clustered training data while keeping its main characteristics. However, as opposed to the DQGAN simulated on a classical computer, the DQGAN\textsubscript{Q} does not achieve to produce the full range of training data.

\end{document}